\documentclass[twocolumn,showpacs,pre,aps,longbibliography,superscriptaddress]{revtex4-1}
\usepackage[bookmarks=false,linkcolor=blue,urlcolor=blue,colorlinks,citecolor=blue]{hyperref}
\pdfoutput=1 

\usepackage{amsmath}
\usepackage{mathtools}
\usepackage{amsthm}
\usepackage{amssymb}
\usepackage{graphicx}
\usepackage[shortlabels]{enumitem}
\usepackage[retainorgcmds]{IEEEtrantools}

\renewcommand{\H}{\mathcal{H}}
\newcommand{\V}{\mathcal{V}}
\newcommand{\B}{\mathcal{B}}
\newcommand{\E}{\mathcal{E}}
\newcommand{\X}{\mathcal{X}}
\renewcommand{\O}{\mathcal{O}}

\newcommand{\inter}{1}
\newcommand{\intra}{0}
\newcommand{\restintra}{{\Lambda,\intra}}
\newcommand{\block}{\text{b}}
\newcommand{\restblock}{{\Lambda,\block}}

\newcommand{\Ham}{\mathrm{H}} 
\newcommand{\K}{\mathcal{K}} 
\newcommand{\Prob}{\mathcal{P}} 

\newcommand{\piecewise}[1]{\left\{\begin{array}{ll}#1\end{array}\right.}
\newcommand{\abs}[1]{\left| #1 \right|} 

\newcommand{\avg}[1]{\left\langle{#1}\right\rangle}
\newcommand{\nn}[1]{\left\langle {#1}\right\rangle}

\newcommand{\matrixel}[3]{\left\langle #1 \vphantom{#2#3} \right| #2 \left| #3 \vphantom{#1#2} \right\rangle} 

\newcommand{\aref}[1]{Appendix \ref{appendix:#1}}

\newtheorem*{lemma}{Lemma}
\newtheorem{corollary}{Corollary}
\newtheorem{proposition}{Proposition}

\begin{document}
	\title{Optimal Renormalization Group Transformation from Information Theory}

	\author{Patrick M. Lenggenhager}
		\affiliation{Institute for Theoretical Physics, ETH Zurich, 8093 Zurich, Switzerland}
	\author{Doruk Efe G{\" o}kmen}
		\affiliation{Institute for Theoretical Physics, ETH Zurich, 8093 Zurich, Switzerland}
	\author{Zohar Ringel}
		\affiliation{Racah Institute of Physics, The Hebrew University of Jerusalem, Jerusalem 9190401, Israel}
	\author{Sebastian D. Huber}
		\affiliation{Institute for Theoretical Physics, ETH Zurich, 8093 Zurich, Switzerland}
	\author{Maciej Koch-Janusz}
		\affiliation{Institute for Theoretical Physics, ETH Zurich, 8093 Zurich, Switzerland}
	

	\begin{abstract}
		Recently a novel real-space RG algorithm was introduced, identifying the relevant degrees of freedom of a system by maximizing an information-theoretic quantity, the real-space mutual information (RSMI), with machine learning methods. Motivated by this, we investigate the information theoretic properties of coarse-graining procedures, for both translationally invariant and disordered systems.
		We prove that a perfect RSMI coarse-graining does not increase the range of interactions in the renormalized Hamiltonian, and, for disordered systems, suppresses generation of correlations in the renormalized disorder distribution, being in this sense \emph{optimal}. We empirically verify decay of those measures of  complexity, as a function of information retained by the RG, on the examples of arbitrary coarse-grainings of the clean and random Ising chain. The results establish a direct and quantifiable connection between properties of RG viewed as a compression scheme, and those of physical objects i.e.~Hamiltonians and disorder distributions. We also study the effect of constraints on the number and type of coarse-grained degrees of freedom on a generic RG procedure. 
	\end{abstract}

	\maketitle
	
	
	\section{Introduction}
	
	The conceptual relations between physics and information theory date back to the very earliest days of statistical mechanics; they include the pioneering work of Boltzmann and Gibbs on entropy~\cite{Boltzmann1877,Gibbs1902}, finding its direct counterpart in Shannon's information entropy~\cite{Shannon1948}, and investigations of Szilard and Landauer~\cite{Szilard1929,Landauer1961}. In the quantum regime research initially focused on foundational challenges posed by the notion of entanglement, but soon gave rise to the wide discipline of quantum information theory~\cite{Bennett1998}, whose more practical aspects include quantum algorithms and computation.
	
	In recent years there has been a renewed interest in applying the formalism and tools of information theory to fundamental problems of theoretical physics. The motivation mainly comes from two, not entirely unrelated, directions. On the one hand the high-energy community is actively investigating the idea of holography in quantum field theories~\cite{tHooft1993,Susskind1995,Bousso2002}, originally inspired by black-hole thermodynamics. On the other hand in condensed matter theory there is a growing appreciation of the role of the entanglement \emph{structure} of quantum wave functions in determining the physical properties of the system. This is exemplified by the short- and long-range entanglement distinguishing the symmetry protected  topological phases~\cite{Kitaev2006,Levin2006,Chen2010} (e.g.\ topological insulators) from genuine, fractionalized topological orders (e.g.\ Fractional Quantum Hall states). The conceptual advances led also to \emph{constructive} developments in the form of new ans{\"a}tze for wave functions (MPS~\cite{Oestlund1995}, MERA~\cite{Vidal2008}) and numerical algorithms (DMRG~\cite{White1992}, NQS~\cite{Carleo602}). 
	
	The focus of this work is on the renormalization group (RG). One of the conceptually most profound developments in theoretical physics, in particular condensed matter theory, it provides -- beyond more direct applications -- a theoretical foundation for the notion of universality~\cite{Kadanoff1966,Wilson1974,Wilson1975,Fisher1998,Efrati2014}.
	The possible connections of RG to information theory have been explored in a number of works \cite{PhysRevD.54.5163,PhysRevLett.81.3587,Apenko2012,Machta604,Beny2015a, Beny2015b,Koch-Janusz2018} in both classical and quantum settings. In particular, some of the present authors introduced a numerical algorithm
	for real-space RG of classical statistical systems \cite{Koch-Janusz2018}, based on the characterization of relevant degrees of freedom supported in a spatial block as the ones sharing the most mutual information with the environment of the block. The algorithm employs machine learning techniques to extract those degrees of freedom and combines it with an iterative sampling scheme of Monte Carlo RG \cite{PhysRevLett.37.461,Swendsen1979}, though, in a crucial difference, the form of the RG coarse-graining rule is not given, but rather \emph{learned}. Strikingly, the coarse-graining rules discovered by the algorithm for the test systems were in an operational sense optimal~\cite{optimal-misc}: they ignored irrelevant short-scale noise and they result in simple effective Hamiltonians or match non-trivial analytical results.
	
	The above suggests, that real-space RG can be universally defined in terms of information theory, rather than based on problem-specific physical intuition. Here we develop a theoretical foundation inspired by, and underlying those numerical results. We show they were not accidental, but rather a consequence of general principles. To this end we prove that a \emph{perfect}, full-RSMI-retaining coarse-graining of a finite-range Hamiltonian does not increase the range of interactions in the renormalized Hamiltonian, in any dimension.
	We then study analytically generic coarse-grainings and the effective Hamiltonian they define, as a function of the real-space mutual information with the environment (RSMI) retained. For the example of the Ising chain we perturbatively derive all the couplings in the renormalized Hamiltonian resulting from, and RSMI captured by, an \emph{arbitrary} coarse-graining and show monotonic decay of the higher-order and/or long-range terms with increased RSMI. 
	
	Those properties also hold in the presence of disorder. We further prove that perfect RSMI-maximizing coarse grainings are stable to local changes in disorder realizations and suppress generation of correlations in the renormalized disorder distribution. Using the solvable example of random dilute Ising chain, we study the properties of the renormalized disorder distribution induced by an arbitrary RG procedure, and show decay of statistical measures of correlation in that distribution as a function of the RSMI retained.
	
	We also theoretically investigate the effects imposed by constraints on the number and type of coarse-grained variables, which can make the loss of part of relevant information inevitable. 
	 We construct simple toy models providing intuitive understanding of our results.
	
	Our results establish a direct link between compression theory intuitions behind introduction of RSMI~\cite{infbottle1}, and physical properties of the renormalized Hamiltonian/disorder distribution. They strongly support RSMI-maximization as a model-independent variational principle defining the \emph{optimal} RG coarse-graining. In contrast to fixed schemes, this RG approach is, by construction, informed by the physics of the system considered, including the position in the phase diagram. This could allow application of  RG schemes to systems, for which they are currently not known, avoiding many of the pitfalls befalling fixed RG transformations \cite{vanEnter1993,Kennedy1997}.

	The paper is organized as follows: in Sec.~\ref{sec:RSMI} the information theoretic formalism and the RSMI algorithm are reviewed, in Sec.~\ref{sec:effHam} we prove that a RSMI-perfect RG does not generate longer-range interactions.
	In Sec.~\ref{sec:Ising} we investigate the renormalized Hamiltonian as a function of the information retained, on the example of arbitrary coarse-grainings of the 1D Ising model. In Sec.~\ref{sec:toymodels} we study the effect of constraints on the number and type of coarse-grained degrees of freedom on a generic RG procedure. We introduce toy models explaining the differences in optimal coarse-grainings in 1D and 2D.
	In Sec.~\ref{sec:disorder} we extend the analysis to disordered systems. We prove RSMI-perfect RG does not generate correlations in disorder. We study properties of the renormalized disorder as a function of the information retained for arbitrary coarse-grainings of the random dilute Ising chain.
	Finally, in Sec.~\ref{sec:conclusions} we discuss implications of the results, generalizations and open questions. Appendices give details of the proofs, derivations of the statements in the main text, and additional information.

	
\section{The RSMI algorithm}\label{sec:RSMI}

The real-space mutual information (RSMI) algorithm is defined in the context of real-space RG, originally introduced by Kadanoff for lattice models~\cite{Kadanoff1966}.
The goal of real-space RG~\cite{Efrati2014} is to coarse-grain a given set of degrees of freedom $\X$ in position space in order to 
integrate out short-range fluctuations and retain only long-range correlations, and in so doing to construct an effective theory. An iterative application of this procedure should result in recursive relations between coupling constants of the Hamiltonian at successive RG steps -- those are the RG flow equations formalising the relationship between effective theories at different length scales.

Consider a generic system with real-space degrees of freedom $\X$ described by the Hamiltonian $\Ham[\X]$  and a canonical partition function:
	\begin{equation}
		Z = \sum_\X e^{-\beta\Ham[\X]} \equiv \sum_\X e^{-\K[\X]}
	\end{equation}
	with the inverse temperature $\beta=1/k_\mathrm{B}T$ and the reduced Hamiltonian $\K:=-\beta\Ham$. Equivalently, the  system is specified by  a probability measure:
	\begin{equation}
	P(\X) = \frac{1}{Z}e^{\K[\X]}.
	\end{equation}
The coarse-graining transformation $\X\to\X'$  between the set of the original degrees of freedom and a (smaller) set of new degrees of freedom  is given by a conditional probability distribution $P_\Lambda(\X'|\X)$, where $\Lambda$ is a set of parameters completely specifying the rule (note, that the rule can be totally deterministic, in which case $P_{\Lambda}$ is a delta-function). The probability measure of the coarse-grained system is then:
\begin{equation}
P(\X') = \sum_\X P_\Lambda(\X'|\X)P(\X).
\label{eq:RG_IT}
\end{equation}
If $P(\X')$ is (or at least can be approximated by) a Gibbs measure, then the requirement to correctly reproduce thermodynamics enforces $Z'=Z$ and a renormalized Hamiltonian $\Ham'[\X']$ in the new variables $\X'$ can be defined implicitly via:
		\begin{equation}
		e^{\K'[\X']} = \sum_\X P_\Lambda(\X'|\X)e^{\K[\X]}.
		\label{eq:RG_IT_Ham}
		\end{equation}
The procedure is often implemented in the form of block RG~\cite{Efrati2014, Niemeyer1974}. This corresponds to a factorization of the conditional probability distribution into independent contributions from equivalent (assuming translation invariance) blocks $\V\subset\X$:
	\begin{equation}
		P(\X'|\X) = \prod_{j=1}^n P_\Lambda(\H_j|\V_j),
		\label{eq:rg_rule}
	\end{equation}
where $\{\V_j\}_{j=1}^n$ and $\{\H_j\}_{j=1}^n$ are partitions of $\X$ and $\X'$, respectively, and $P_{\Lambda}$ now defines the coarse-graining of a single block (and therefore  $\Lambda$ contains substantially fewer parameters). Concrete examples of such $P_{\Lambda}$ include the standard ``decimation" or ``majority-rule" transformations [see Eqs.(\ref{eq:decimation_rule},\ref{eq:majority_rule})].
	
Not every choice of $P_{\Lambda}$ is \emph{physically} meaningful. It should at least be consistent with the
symmetries of the system under consideration, for instance. This is, however,  not sufficient in practice. While it may be difficult to formulate a concise criterion for the choice of the coarse-graining transformation it is clear that in order to derive the recursive RG equations the effective Hamiltonian cannot proliferate new couplings at each step. If there is to be a chance of analytical control over the procedure, the interactions in the effective Hamiltonian should  be tractable (short-ranged, for instance). That is to say, if one chooses the ``correct" degrees of freedom to describe the system, the resulting theory should be ``simple". Numerous examples of failure to achieve this can be found in the literature \cite{vanEnter1993,Kennedy1997}, and include cases as simple as decimation of the Ising model in 2D. Implicit in this is the notion that there does not exist a single RG transformation which does the job, but rather the transformation should be designed for the problem at hand~\cite{Fisher:1982yc}.

	\begin{figure}
		\centering
		\includegraphics{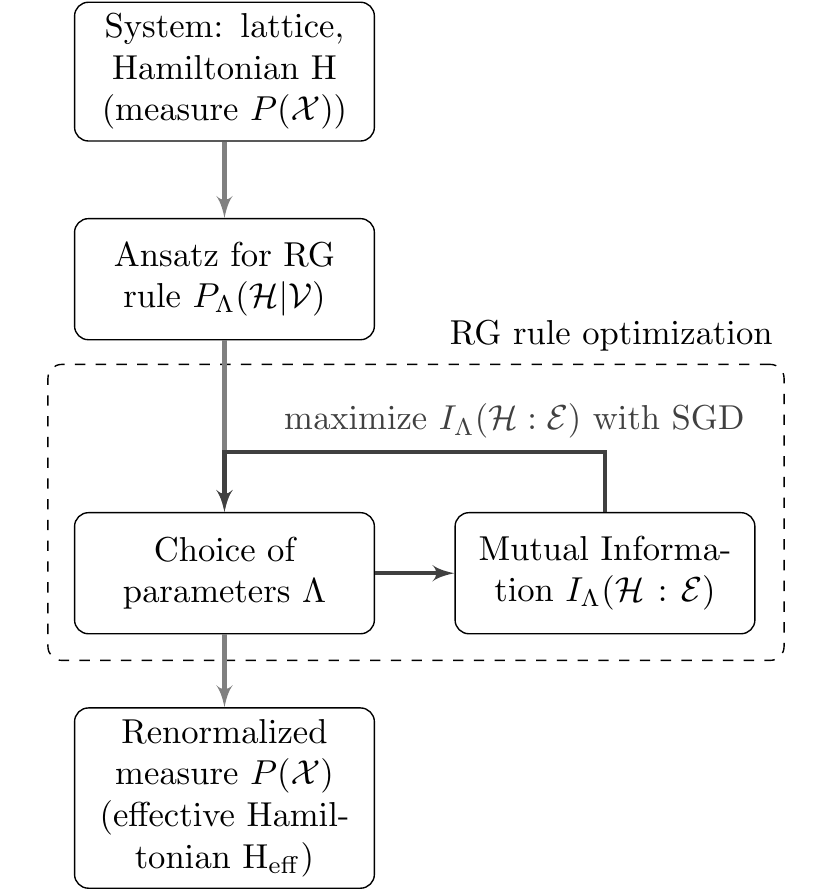}
		\caption{
			Flow-diagram of the RSMI algorithm~\cite{Koch-Janusz2018}.
			Given a lattice and Hamiltonian $\Ham$ (or, in practice, given Monte Carlo samples) an RBM-ansatz for the RG rule is optimized 
			by maximizing the mutual information between the new degrees of freedom $\H$ and the environment $\E$  of the original ones using Stochastic Gradient Descent (SGD).
			The trained $P_\Lambda(\H|\V)$  is used to define a new effective measure and Hamiltonian $\Ham_\text{eff}$.
		}
		\label{fig:rsmimalg}
	\end{figure}

Recently, some of us proposed the maximization of the real-space mutual information (introduced below) as a criterion
for a \emph{physically meaningful} RG transformation~\cite{Koch-Janusz2018}. The idea behind it is that the effective block degrees of freedom, 
in whose terms the long-wavelength theory is simple, are those which retain the most of the information (already present in the block) about long-wavelength properties of the system.
This informally introduced ``information" can be formalized by the following construction. Consider a single block $\V$ at a time and divide the system into four regions $\X=\V\cup\B\cup\E\cup\O$: the visibles (i.e. the block) $\V$, the buffer $\B$, the environment $\E$ and the remaining outer part of the system $\O$ (which is only introduced for algorithmic reasons, conceptually the environment $\E$ could also contain this part). Fig.(\ref{fig:system1Dising}) depicts this decomposition in the case of a 1D spin model, but it trivially generalizes to any dimension. The real-space mutual information between the new (coarse-grained) degrees of freedom $\H$ and the environment $\E$ of the original ones (i.e. of the block) is then defined as:
	\begin{equation}
		I_\Lambda(\H:\E)
		=\sum_{\H,\E}P_\Lambda(\E,\H)\log\left(\frac{P_\Lambda(\E,\H)}{P_\Lambda(\H)P(\E)}\right)
		\label{eq:MI_def}
	\end{equation}
where $P_\Lambda(\E,\H)$ and $P_\Lambda(\H)$ are marginal distributions of $P_{\Lambda}(\H,\X)=P_\Lambda(\H|\V)P(\X)$. Thus $I_\Lambda(\H:\E)$  is the standard mutual information between the random variables $\H$ and $\E$. Exclusion of the buffer $\B$ (in contrast to other adaptive schemes, see for instance \cite{Brandt2001}), generally of linear extent comparable to $\V$, is of fundamental importance: it filters out short-range correlations, leaving only the long-range contributions to $I_\Lambda(\H:\E)$.

The RSMI satisfies the following bounds (see also Appendix \ref{appendix:MI}):
	\begin{IEEEeqnarray}{rCl}
		0 \leq I_\Lambda(\H:\E) &\leq& H(\H),
		\label{eq:MIbound_entropy}\\
		I_\Lambda(\H:\E) &\leq& I(\V:\E),\label{eq:MIbound_visibles}
	\end{IEEEeqnarray}
where $H(\H)$ denotes the information entropy of $\H$ and $I(\V:\E)$ is the mutual information of the visibles with the environment. The optimization algorithm starts with a set of samples drawn from $P(\X)$ and a differentiable ansatz for $P_{\Lambda}(\H|\V)$, which in Ref.~\cite{Koch-Janusz2018} takes the form of a Restricted Boltzmann Machine (RBM), parametrized by $\Lambda$ (see Appendix~\ref{appendix:RBMansatz}), and updates the parameters using a (stochastic) gradient descent procedure. The cost function to be maximized is precisely $I_\Lambda(\H:\E)$, which in the course of the training is increased towards the value of $I(\V:\E)$. The iterative procedure is shown in Fig.~\ref{fig:rsmimalg}. Using the trained $P_\Lambda(\H |\V)$ the original set of samples drawn from $P(\X)$ can be coarse-grained and the full procedure re-computed for a subsequent RG step.

		\begin{figure}
			\centering
			\includegraphics{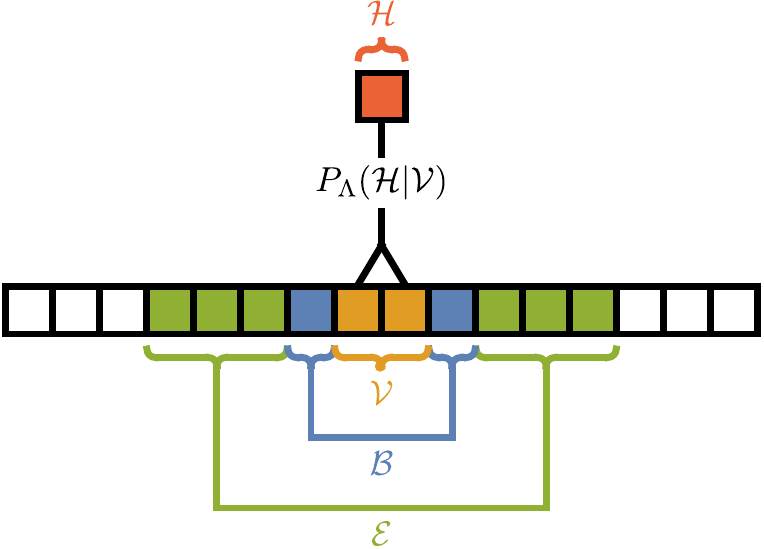}
			\caption{
				Schematic decomposition of the system for the purpose of defining the mutual
				information $I_\Lambda(\H:\E)$ (in 1D, for concreteness). The full system is partitioned into blocks of visibles $\V$ (yellow)
				embedded into a buffer $\B$ (blue) and surrounded by the environment $\E$ (green). The remaining part of the system is denoted by $\O$ in the main text. The conditional probability distribution $P_\Lambda(\H|\V)$ couples $\V$  to the hiddens $\H$ (red).
			}
			\label{fig:system1Dising}
		\end{figure}


\section{Optimality: the measure and the effective Hamiltonian}\label{sec:effHam}

In what sense is the RSMI coarse-graining \emph{optimal}? By construction, the scheme preserves as much information about long-range properties of the system as possible, and thus, when viewed as a compression of the \emph {relevant} information in $\V$ into $\H$, it is information-theoretically optimal \cite{bottleneck-misc}. We will show that this well-defined but abstract notion implies \emph{physical} ``simplicity" of the renormalized Hamiltonians. The latter, though intuitively clear and operationally useful, may be difficult to define unambigously. We will, therefore, examine natural measures of Hamiltonian complexity, and show they \emph{all} decay with increased MI, also for disordered systems. It will also prove useful to approach this problem at the level of properties of the probability measure (which is the fundamental object the RSMI algorithm works with).

Consider first the following setup: given  a 1D system with a short-ranged Hamiltonian introduce a coarse-graining $\{\V_j \}$, with a block size chosen so that the Hamiltonian is nearest-neighbour with respect to the blocks.  Let us choose an arbitrary block $\V_0$, denote its immediate neighbours $\V_{\pm 1}$ as the buffer $\B$, and all the remaining blocks $\{\V_{j<-1}  \}$ and $\{\V_{j>1}  \}$ as the environment $\E_0$, or in more detail, as left- and right-environment $\E_{L/R}(\V_0) $, respectively. 
Assume now that $\H_0$, the coarse-grained variable for $\V_0$, is constructed so that $I(\H_0 : \E_0) = I(\V_0 : \E_0)$, i.e.~the coarse-grained variable retains \emph{all of the} information which the original block $\V_0$ contained about the environment, and thus about any long-wavelength physics. The following then holds true (proof in Appendix \ref{appendix:proofRange}): 
\begin{proposition}
Let $I(\H_0 : \E_0) = I(\V_0 : \E_0)$.  Then the probability measure on the coarse-grained variables $P( \{\H_j \})$ obeys the factorization property: 
\begin{equation}\label{eq:factorization1}
P(\H_{j \leq -2},\H_{j \geq2} | \H_0) = P(\H_{j\leq -2} | \H_0)\cdot P(\H_{j\geq 2} | \H_0),
\end{equation}
where in the conditional probabilties the buffer variables (i.e.~the neighbours $\H_{\pm 1} $ of $\H_0$) have been integrated out. In other words, for a fixed $\H_0$ the probabilities of its environments $\E_{L/R}(\H_0)$ are independent of each other. 
\end{proposition}

An immediate consequence of the above is:
\begin{corollary}
	The effective Hamiltonian does not contain terms directly coupling $\E_{L}(\H_0)$  and $\E_{R}(\H_0)$.
\end{corollary}
This is because the factorization Eq. (\ref{eq:factorization1}) implies:
\begin{gather}\label{eq:factorization2}
E(\H_{j \leq -2}, \H_0,\H_{j \geq2}) \propto \log\left[ P(\H_{j \leq -2}, \H_0,\H_{j \geq2}) \right] \\
 =  E(\E_{L}, \H_0) + E(\E_{R}, \H_0) + E(\H_{0}) \nonumber
\end{gather}

Since the variables $\E_{L}(\H_{0})\coloneqq\{\H\}_{j<-1}$ and $\E_{R}(\H_{0}) \coloneqq \{\H \}_{j>1}$ are decoupled \emph{after} integrating out the buffer $\H_{\pm 1}$ there generically would not have been any longer-range interaction (in particular: next-nearest neighbour) involving $\H_{\pm 1}$ in the renormalized Hamiltonian, or the measure would not factorize. Since the choice of $\V_0$, $\E_0$ was arbitrary in the first place, we have:
\begin{corollary} For a finite range Hamiltonian, if $I(\H_0 : \E_0) = I(\V_0 : \E_0)$, the RSMI coarse-graining is guaranteed not to increase the range of interactions.
\end{corollary}

This generalizes, under very mild additional assumptions, to any dimension $D$. Taking a coarse-graining with blocks sufficiently large to make the short-ranged Hamiltonian nearest-neighbour, and under the assumption of full information capture, we repeat the above reasoning, conditioning on -- instead of a single arbitrary variable $\H_0$ -- a hyperplane of dimension $D-1$, separating the coarse grained variables $\{\H_j  \}$ into two disconnected sets, to show that no longer-ranged interactions across the hyperplane can exist. Since the choice of hyperplane is arbitrary, the effective Hamiltonian is nearest-neighbour, as the original one was (see Appendix \ref{appendix:proofRange}). A \emph{perfect} RSMI scheme does not, therefore, increase the range of a short-ranged Hamiltonian. i.e. its complexity.

While a strong property, the above results appear to have one serious shortcoming: for a generic physical system and coarse-graining scheme \emph{it may not be possible} to satisfy the assumption $I(\H_0 : \E_0) = I(\V_0 : \E_0)$, which is a strict upper bound on $I(\H_0 : \E_0)$, \emph{for any RG rule}.
This is due to the fact that the block size, as well as the number and character (Ising spin, Potts spin, ...) of coarse-grained variables are usually chosen \emph{a priori}, and given those constraints a solution $P_{\Lambda}(\mathcal{H}|\mathcal{V})$ satisfying $I(\H_0 : \E_0) = I(\V_0 : \E_0)$ is not mathematically guaranteed to exist (see Sec.~\ref{sec:toymodels} for examples). 
This, however, is only a superficial problem. First, Proposition 1 is a sufficient, and not a necessary condition. Much more importantly, the RSMI prescription is a variational principle. If the physics of the problem and constraints imposed preclude the existence of a ``perfect" solution, as is usually the case, the maximization of RSMI still yields \emph{the best possible one}, given the conditions. A mathematical proof of this statement requires establishing decay of some measures of the effective Hamiltonian complexity (such as range and the ones we consider below) as a function of the mutual information. In the absence of such rigorous result, in what follows we instead study, analytically and numerically, tractable models and verify that this decay indeed holds empirically i.e.~the more mutual information RG rule captures, the smaller complexity of effective Hamiltonian. Furthemore, we show this also holds in the presence of disorder (see Sec.~\ref{sec:disorder}).

We now investigate a realistic setup, in which the RSMI is maximized under the constraint of number and type of coarse-grained degrees of freedom. Additionally, since the RG rule is optimized iteratively, we study the approach to the optimal solution via the properties of the renormalized Hamiltonian defined by the RG rule at any stage of the procedure. We briefly review how the effective Hamiltonian can be expressed by appropriate cumulant expansion~\cite{Niemeyer1974} (though the RSMI algorithm deals with probability measure as the basic object, and at no point computes the Hamiltonian, the Hamiltonian is more interpretable physically) and we apply this machinery to the Ising chain with and without disorder.


\section{Soluble example: Arbitrary RG transformations of the clean 1D Ising Model}\label{sec:Ising}
	
To investigate the relationship between the renormalized Hamiltonian and the real-space mutual information for practical coarse-graining procedures, we consider the example of the one-dimensional Ising model with nearest-neighbor interactions and periodic boundary conditions. We deliberately use this simple model, since it allows to analytically derive properties not only of the optimal RG procedure (which we do first), but also those of \emph{arbitrary} coarse-grainings: both the effective Hamiltonian and the amount of RSMI captured can be calculated explicitly and without any arbitrary truncations to establish the relation beween them.
The Ising Hamiltonian reads:
	\begin{equation}
	\K[\X] = K\sum_{i = 1}^N x_ix_{i+1},
	\label{eq:Ham_Ising}
	\end{equation}
with $x_i=\pm 1$ collectively denoted by  $\X=\{x_i\}_{i=1}^N$ and with $K:=-\beta J$. The sizes of the block, buffer and environment regions, introduced in Sec.\ \ref{sec:RSMI}
are given by $L_\V$, $L_\B$ and $L_\E$.  Accordingly, there are $n=N/L_\V$ blocks.

To best illustrate the results we now specialize to the (typical) case of blocks of	two visible spins $\V=\{v_1,v_2\}$, coarse-grained into a single hidden spin $h$ (computations for general $L_\V$ are analogous).
The RG rule is parametrized by an RBM ansatz:
	\begin{equation}
		P_\Lambda(\H|\V) = \frac{1}{1+e^{-2h\sum_{i}\lambda_i v_i}},
		\label{eq:RG_rule_RBM_ansatz}
	\end{equation}
	with $\Lambda=(\lambda_1,\lambda_2)$ describing the quadratic coupling of visible to hidden spins
	(see \aref{RBMansatz} for discussion of the ansatz). In Fig.(\ref{fig:system1Dising})  the decomposition of the system and the RG rule
	are schematically shown.
		
The standard decimation and the majority rule coarse-graining schemes are given in our language by:

	\begin{equation}
	P_{\textrm{dec}}(h|\{v_1,v_2\}) = \piecewise{1,&h=v_1\\0,&h\neq v_1}, \mbox{\ \ \ \ \ \ \ }
	\label{eq:decimation_rule}
	\end{equation}
	and by:
	\begin{equation}
	P_{\textrm{maj}}(h|\{v_1,v_2\}) = \piecewise{
		1,&v_1=v_2=h\\0,&v_1=v_2\neq h\\\frac{1}{2},&v_1\neq v_2},
	\label{eq:majority_rule}
	\end{equation}
	respectively. They correspond to the choice of $\Lambda_{\textrm{dec}}=(\lambda,0)$ and $\Lambda_{\textrm{maj}}=(\lambda,\lambda)$ in the limit $\lambda \rightarrow\infty$.
	
	For the case of decimation an exact calculation using the transfer matrix approach yields an effective Hamiltonian of the same nearest-neighbour form, albeit with a renormalized coupling constant~\cite{Kramers1941,Onsager1944}:
		\begin{equation}
		K' = \frac{1}{2}\log(\cosh(2K)).
		\label{eq:ren_coupling_dec}
		\end{equation}
	
	For the majority rule, and any other choice of parameters $\Lambda$, the renormalized Hamiltonian cannot be obtained in a closed form, but can still be derived analytically. To this end we split it into two parts~\cite{Niemeyer1974}:
	\begin{equation}
	\K[\X] = \K_\intra[\X] + \K_\inter[\X],
	\label{eq:Ham_decomp_intra_inter}
	\end{equation}
	where $\K_\intra$ contains \emph{intra}-block and $\K_\inter$ \emph{inter}-block terms. Using the cumulant expansion the new Hamiltonian is given perturbatively:
	\begin{equation}
	\K'[\X'] = \log(Z_\intra P_\restintra(\X')) + \sum_{k=0}^\infty\frac{1}{k!}C_k[\X'],
	\label{eq:ren_Ham}
	\end{equation}
	where the cumulants $C_k$ can be expressed in terms of averages of the form
	$\avg{\K_1[\X]^k}_\restintra$, which factorize into averages of operators
	from a single block (see Appendix \ref{sec:app_cumulantexp} for details). The renormalized coupling constants are not	apparent in Eq.(\ref{eq:ren_Ham}).
	In order to identify them we introduce the following canonical form of the Hamiltonian:
	\begin{equation}
	\K'[\X'] = K_0' + \sum_{\{\alpha_\ell\}_{\ell=1}^n}
	K'_{\alpha_1,\alpha_2,\dotsc,\alpha_n}
	\left(\sum_{j=1}^n\prod_{\ell=1}^n(x_{j+\ell}')^{\alpha_\ell}\right),
	\label{eq:Ham_canonical}
	\end{equation}
	with $\alpha_1= 1$ and $\alpha_\ell\in\{0,1\}$ for all $\ell>1$.
	Here, addition of the indices is to be understood modulo $n$ (i.e. with periodic boundary conditions).
	Note that arbitrary orders $k$ of the cumulant expansion $C_k$ contribute to each coupling constant $K'_{\alpha_1,\alpha_2,\dotsc,\alpha_n}$.
	
		\begin{figure}
			\centering
			\includegraphics{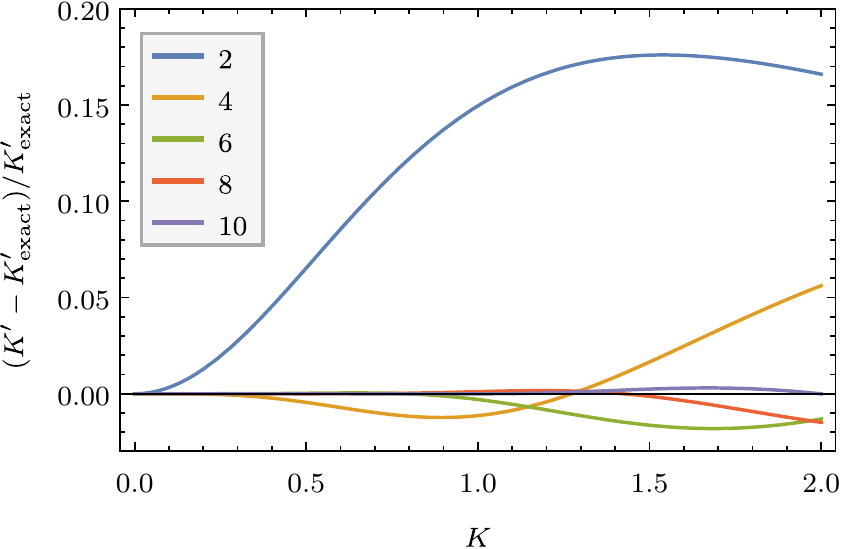}
			\caption{The relative difference between the renormalized NN coupling obtained from the cumulant expansion of the RSMI-favoured solution (decimation), and the nonperturbative result Eq.(\ref{eq:ren_coupling_dec}). The convergence improves with increasing order of cumulant expansion $M$  and lower $K$.}
			\label{fig:compexactdec}
		\end{figure}	
		
				\begin{figure*}[th]
					\centering
					\includegraphics{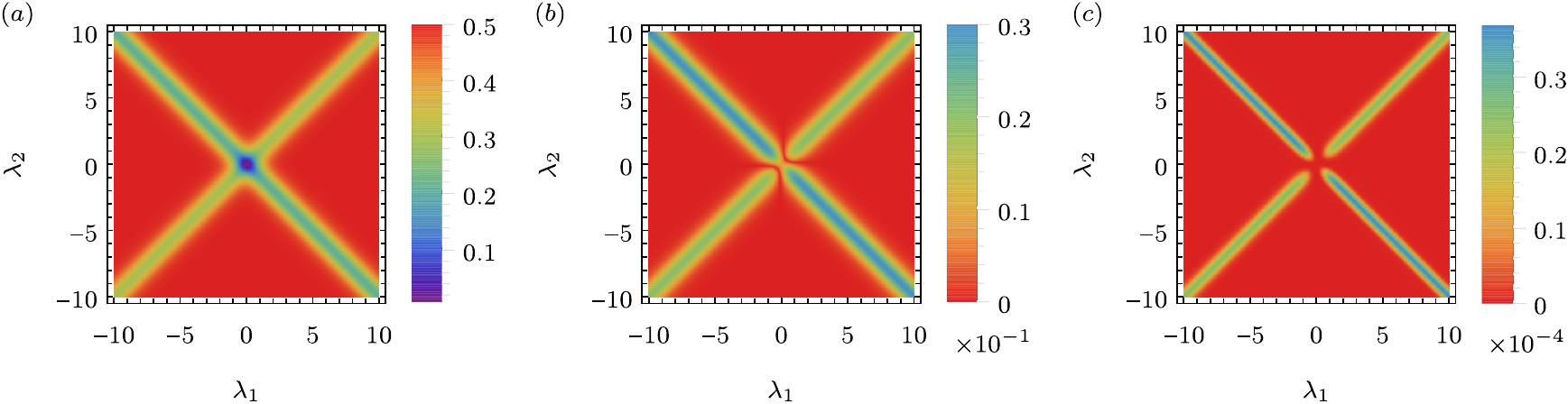}
					\caption{Density plots of (a) the mutual information of the hidden with the environment
						scaled to the mutual information of the visibles with the environment $I_\Lambda(\H:\E)/I(\V:\E)$,
						(b) the ratio of the next-nearest-neighbour (NNN) to the nearest-neighbour (NN) coupling constants $|K_ 2'(2)/K_2'(1)|$ and
						(c) the ratio of the NN  four-point to two-point coupling constants $|K_ 4'(1,1,1)/K_2'(1)|$.
						All three quantities are shown as a function of the parameters of the RG rule $\Lambda=(\lambda_1,\lambda_2)$. Note the inverted color scale in (b) and (c)!
						For large enough $||\Lambda||^2$ a maximum of mutual information corresponds to a minimum of ``rangeness" and ``m-bodyness", and vice versa. See the main text and Appendix \ref{appendix:lv2plots} for details.
					}
					\label{fig:densityplotlv2}
				\end{figure*}	

   In the example of the Ising model the only non-vanishing averages contributing to the cumulants are:
	\begin{IEEEeqnarray}{rCl}
		\IEEEyesnumber\label{eq:effbparams}
		\avg{v_1}_\restblock[h] &=:& a_1h,\IEEEyessubnumber\\
		\avg{v_{2}}_\restblock[h] &=:& a_2h,\IEEEyessubnumber\\
		\avg{v_1v_{2}}_\restblock[h] &=:& b,\IEEEyessubnumber
	\end{IEEEeqnarray}
	with the effective block-parameters $a_1$, $a_2$, $b$ independent of the coarse-grained variable $h$ and functions
	of $\Lambda$ and $K$ only, whose closed form expressions can easily be found (see Appendix \ref{app:Ising}).
	Consequently, the averages  $\avg{\K_\inter^k}_\restintra$, and thus also the Hamiltonian $\K'$, are polynomials in the new degrees of freedom $\X'$, the reduced temperature $K$ and the block parameters,  which gives rise to Eq.(\ref{eq:Ham_canonical}).
	In practice the cumulant expansion is terminated at a finite order $M$, which results in an expansion of $\K'$ and thus of each coupling constant 
	$K'_{\alpha_1,\alpha_2,\dotsc,\alpha_n}$ up to that order in $K$.
	All the information about the RG rule (except for the size of $\H$, which is fixed at the outset) is contained in the dependence of the effective block-parameters on $\Lambda$ (and on $N$, $K$).

	Expressing the moments $\avg{\K_1^k}_\restblock$ appearing in the cumulant
	expansion in terms of the new variables $\X'$ is a combinatorial problem.
	Each term in $\K_1$ couples spins from neighboring blocks $j$ and
	$j+1$, so that:
	\begin{equation}
	\K_1[\X]^k = K^k\sum_{j_1,\dotsc,j_k=1}^n\prod_{\ell=1}^k x_{2j_\ell}x_{2j_\ell +1}.
	\label{eq:K1k}
	\end{equation}
	The average of each summand factorizes into contributions from each block, whose value [see Eq.(\ref{eq:effbparams})] is determined by the arrangement of $j_1,\dotsc,j_k$. Thus, the calculation is reduced to finding and grouping all equivalent (under the fact that for Ising variables $x_j^2=1$) configurations $(j_1,\dotsc,j_k)$.
	Bringing the resulting polynomial in canonical form (\ref{eq:Ham_canonical}) is an inverse problem and is solved
	by recursively eliminating non-canonical terms. For a given $M$ we can thus finally arrive at expression of coupling constants $K'_{\alpha_1,\alpha_2,\dotsc,\alpha_n}$ as functions of $\Lambda$ and $K$ (see Appendix \ref{app:Ising} for details).

	We are now in a position to examine the effective Hamiltonian obtained by applying the RSMI-maximization procedure  Fig.(\ref{fig:rsmimalg}) to the model Eq.(\ref{eq:Ham_Ising}). Anticipating the results in Fig.(\ref{fig:densityplotlv2}), in Fig.(\ref{fig:compexactdec}) we compare, for varying $K$ and order of cumulant expansion $M$,  the renormalized nearest-neighbour (NN) coupling obtained in the RSMI-favoured solution with the exact, nonperturbative one Eq.(\ref{eq:ren_coupling_dec}) [which we refer to as ``exact decimation'']. The two results converge with increasing $M$, and the convergence is faster for weak coupling/higher temperatures, which is unsurprising since the cumulant expansion is in powers of $K$. We emphasize again that the RSMI algorithm itself works on the level of the probability measure, and at no point does it compute the effective Hamiltonian. It is only when we want to examine the renormalized Hamiltonian which the converged -- in the sense of saturating the mutual information during optimization of the $\Lambda$ parameters -- RSMI solution corresponds to, that we are performing the cumulant expansion.
	
		\begin{figure*}
			\centering
			\includegraphics{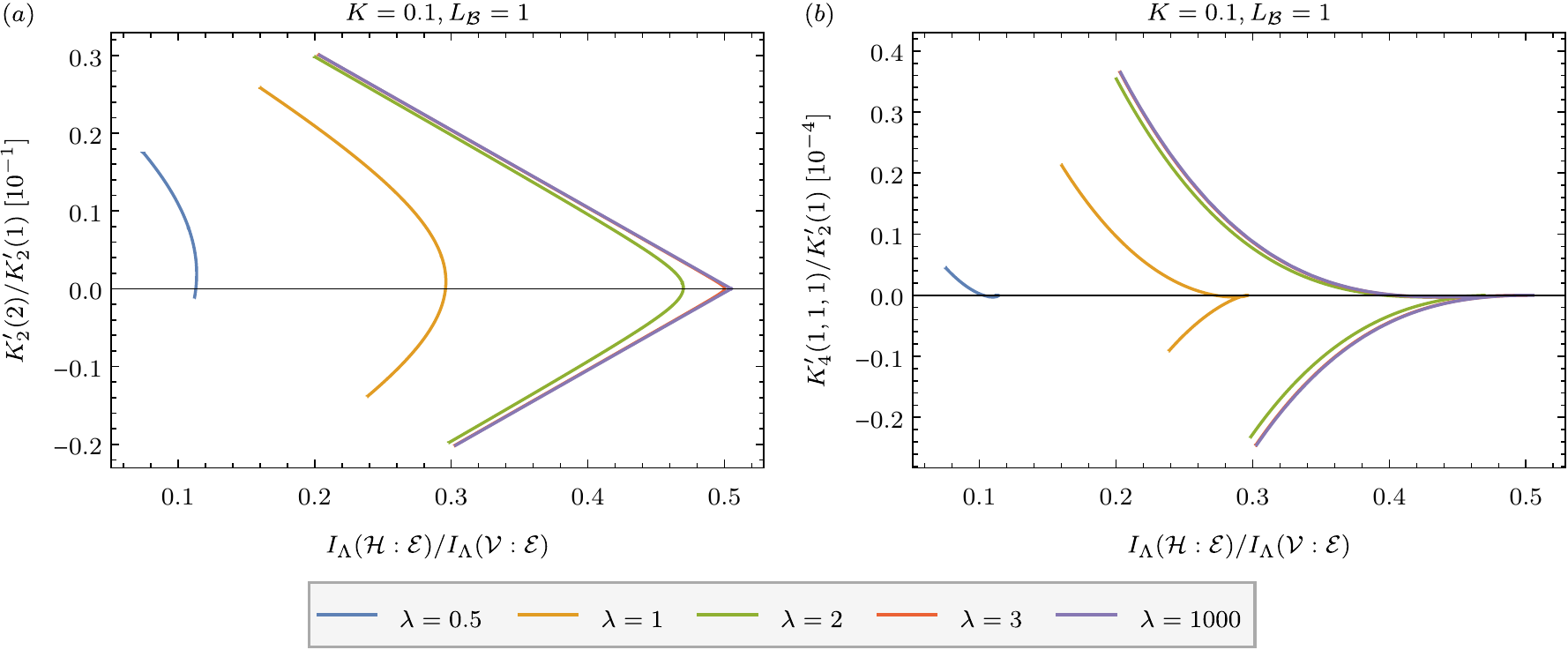}
			\caption{
				The two proxy measures of complexity of the renormalized Hamiltonian discussed in the text are shown against mutual information retained: the ``rangeness" i.e. the ratio of the NNN to the NN coupling constants, and the "m-bodyness" i.e. the ratio of the NN four-point to two-point coupling constants. The mutual information is scaled to the total mutual information the block $\V$ shares with the environment. The curves are obtained by parametrizing the RG rule as
				$\lambda(\cos(\theta),\sin(\theta))$ and varying $\theta\in[0,\pi]$ for different magnitude of $\lambda$. In the physically relevant limit of large $\lambda$ the maximum of mutual information corresponds to a minimum of ``rangeness" and ``m-bodyness". The plots are discussed in more detail in Appendix \ref{appendix:lv2plots}.
			}
			\label{fig:rangenesslv2}
		\end{figure*}
	
	Since ``exact decimation" leads to a strictly NN effective Hamiltonian in the 1D Ising case, and since perturbatively the RSMI-favoured solution converges to the decimation value for the NN coupling, it is instructive to inspect the behaviour of the $m$-body couplings in the effective Hamiltonian for larger order $m$.
	Denoting the $m$-spin coupling with distances $\ell_1,\ell_2,\dotsc,\ell_m$ between the spins by $K_m(\ell_1,\ell_2,\dotsc,\ell_m)$, with $K_m(\ell)$ short for
	$K_m(\ell,\ell,\dotsc,\ell)$, we observe that, in the limit of weak coupling (small $K$), both $\ell\mapsto K_2(\ell)$, i.e. arbitrary range two-body interactions, as well as $m\mapsto \abs{K_m(1)}$, i.e. arbitrary order NN-interactions, decay exponentially. This is shown in Figs.(\ref{fig:appendix:k2decay}) and (\ref{fig:appendix:kmdecay}) in Appendix \ref{app:Ising}.
	The decay length is characterized by $K_2(2)/K_2(1)$ and $K_m(1)/K_2(1)$, respectively. Thus, the RSMI approach indeed converges to the ``exact decimation" in this case, which is known to be the optimal choice. 
	
	To further strengthen the link between the amount of RSMI retained and the resulting properties of the effective Hamiltonian we now consider a \emph{generic} coarse-graining, suboptimal from the RSMI perspective (i.e. away from the maximum the RSMI algorithm strives for). To this end we compute the mutual information  $I_\Lambda(\H:\E)$ captured for the Ising model by a general coarse-graining rule Eq.(\ref{eq:RG_rule_RBM_ansatz}) with parameters $\Lambda = (\lambda_1,\lambda_2)$. This calculation can be performed exactly using the transfer matrix method (see \aref{MI_Ising}) and yields a closed form expression Eq.(\ref{eq:d33}).

	Equipped with these results, for an arbitrary coarse-graining defined by a choice of $\Lambda$, we can now compute both the amount of mutual information 
	with the environment retained (RSMI), as well as the effective Hamiltonian generated. In Fig.(\ref{fig:densityplotlv2}a) the amount of information captured is shown as a function of $(\lambda_1,\lambda_2)$, in units of $I(\V:\E)$ (for concreteness, all plots are for $K=0.1$ and a single site buffer: $L_\B=1$). A few observations can be made: the choices of $\Lambda$ retaining more RSMI are not symmetric in $|\lambda_1|$ and $|\lambda_2|$, but instead tend to $(\pm\lambda,0)$ and $(0,\pm\lambda)$ for large enough $|\lambda|$, i.e. they resemble decimation Eq.(\ref{eq:decimation_rule}) [the four plateaux in Fig.(\ref{fig:densityplotlv2}) are not exactly flat, as also examined in Fig.(\ref{fig:rangenesslv2})], as opposed to majority rule Eq.(\ref{eq:majority_rule}) which, in fact, captures the least information. The symmetries of the plot are due to global $\mathbb{Z}_2$ Ising symmetry as well as an additional $\mathbb{Z}_2$ symmetry of the mutual information: correlation and anti-correlation for random variables is equivalent from the point of view of information. Furthermore, the lack of information retained for small $||\Lambda||^2$ is due to the fact that in this case the coarse-graining Eq.(\ref{eq:RG_rule_RBM_ansatz}) only weakly depends on the visible spins and is essentially randomly assigning the value of the hidden spin (i.e. it is dominated by random noise). In other words, it only makes sense to think of Eq.(\ref{eq:RG_rule_RBM_ansatz}) as a  coarse-graining if it strongly depends on the original spins, i.e. for large $||\Lambda||^2$.

	The properties of the corresponding effective Hamiltonians can be understood with the help of Figs.(\ref{fig:densityplotlv2}b) and (\ref{fig:densityplotlv2}c), where the ratio of next-nearest-neighbour (NNN) to NN terms as well as the ratio of NN 4-body to 2-body terms in the effective Hamiltonian are plotted as a function of $\Lambda$ (note the inverted  color scale!). It is apparent that decimation-like choices, which maximize RSMI, result also in vanishing NNN and 4-body terms (and more generally long-range or high-order terms, as discussed previously and shown in Figs.(\ref{fig:appendix:k2decay}) and (\ref{fig:appendix:kmdecay}) in Appendix \ref{app:Ising}). This is examined in more detail in Fig.\ref{fig:rangenesslv2}: trajectories in the parameter space $\Lambda$ are chosen according to $\lambda(\cos(\theta),\sin(\theta))$ with $\theta\in[0,\pi]$, for different magnitudes $|\lambda|$. The ratios in Figs.(\ref{fig:densityplotlv2}b) and (\ref{fig:densityplotlv2}c), which we dubbed ``rangeness" and ``m-bodyness" for brevity, are plotted against the mutual information along the trajectories. The mutual information is maximized for $\theta = 0$ and $\theta=\pi$ and the maximum increases with $\lambda$ (though it saturates: there is little difference between $\lambda=3$ and $\lambda=1000$). Simultaneously, for large enough $|\lambda|$ both ratios in Figs.(\ref{fig:rangenesslv2}b,c) vanish, rendering the effective Hamiltonian two-body and nearest-neighbour. It is now clear how the RSMI maximization results in a decimation coarse-graining for the 1D Ising model. A more detailed discussion of Figs.(\ref{fig:densityplotlv2},\ref{fig:rangenesslv2}) [including asymmetries in Fig.(\ref{fig:rangenesslv2}a) and accidental vanishings in Fig.(\ref{fig:rangenesslv2}b)] can be found in the Appendix \ref{app:Ising}, but it does not change the general picture: maximizing RSMI results in decay of longer-ranged and higher-order terms in the Hamiltonian.
	
				\begin{figure*}[th]
					\centering
					\includegraphics{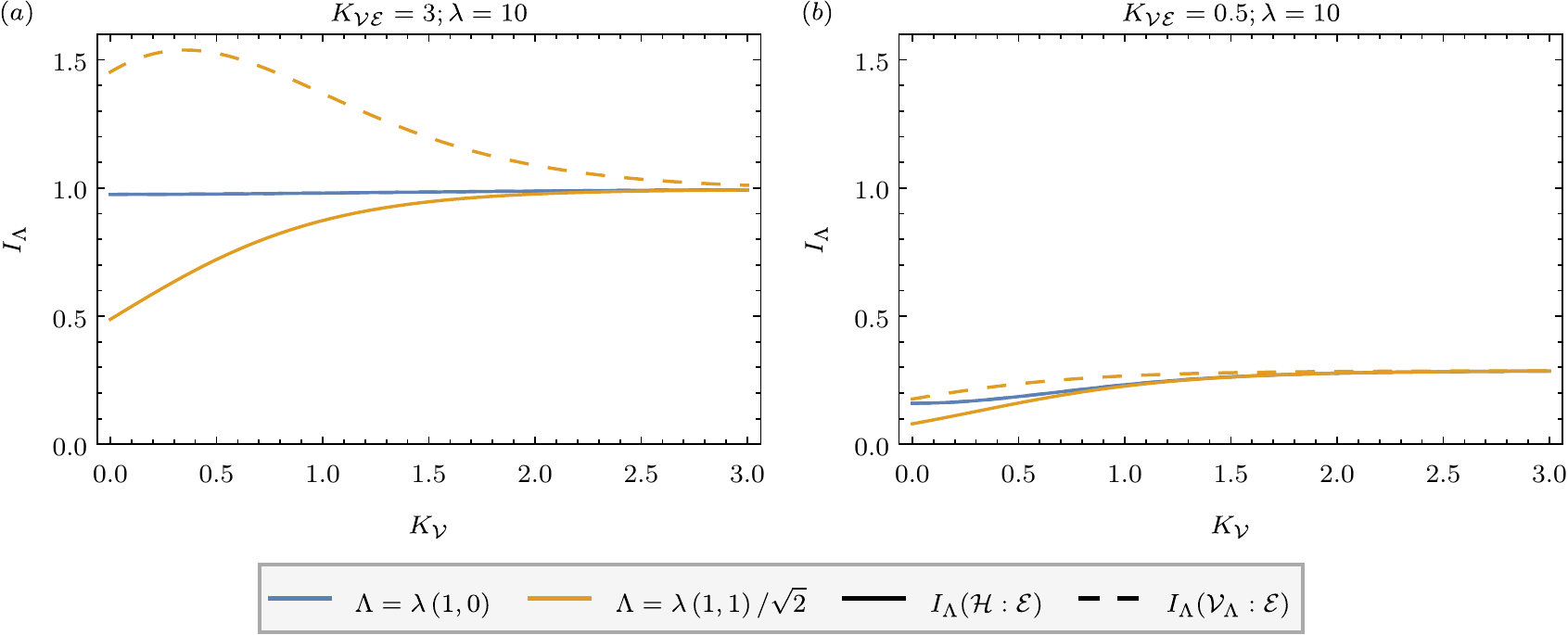}
					\caption{
						The mutual information $I_{\Lambda}(\H:\E)$ and $I(\V_\Lambda :\E)$ for decimation (blue) and majority rule (yellow) procedures in the 1D toy model Eq.(\ref{eq:toymodel_1D_Ham}).
						Two parameter regimes are shown:
						(a) Strong coupling to the environment/ low temperature $K_{\V\E}$ (recall that the coupling constants contain a factor of  $\beta = 1/k_BT$)
						(b) Weak coupling $K_{\V\E}$; note that the absolute value of all mutual informations are lower in this limit.
						The solid lines differ from the dashed lines of the same colour by the mismatch $I(\V_{\Lambda}:\E | \H)$ (see main text).
						In both parameter regimes the dashed blue line exactly coincides with the solid blue line: for decimation procedure the information $I(\V_\Lambda :\E)$ is perfectly encoded into $\H$. Majority rule $I_{\Lambda}(\H:\E)$ is inferior to decimation, even though $I(\V_\Lambda :\E)$ is significantly larger: all that information is lost in encoding. The distinction between the two rules vanishes in the $K_\V\rightarrow\infty$ limit when both visible spins are effectively bound into a single binary variable.
					}
					\label{fig:toymodel1d}
				\end{figure*}
	
	The superiority of decimation over majority rule in our example can be understood intuitively from a physical perspective by considering fluctuations of the original (visible) spins for a fixed (clamped) configuration of the new variables $\X'$.
	In 1D decimation fixes every other spin in $\X$, which prevents all but isolated fluctuations of the remaining degrees of freedom, which are being integrated out in the clamped averages of Eqs.(\ref{eq:app_RGeq_for_cumulants},\ref{eq:app_cumulant_expansion}).
	Consequently, only nearest neighbors in $\X'$ are coupled in the effective Hamiltonian.
	In contrast, the majority-rule fixes a linear combination of the visibles (the average), thereby allowing fluctuations of orthogonal linear combinations.
	These fluctuations can span multiple blocks and thus generate higher order coupling terms. In the following section an alternative, information based intuition is offered, which also explains the difference between the optimal coarse-graining procedures in 1 and 2D.
	
	Finally, we note that the results described above from a static perspective, i.e. considering properties of arbitrary coarse-graining, for a fixed, potentially suboptimal, choice of $\Lambda$, can also be interpeted dynamically. In this sense they would characterize the convergence of the RSMI algorithm of Ref.~\cite{Koch-Janusz2018}  as the $\Lambda$ parameters are iteratively optimized during the training [see Fig.(\ref{fig:rsmimalg})].

	
	\section{The ``shape" of  the coarse-grained variables}\label{sec:toymodels}

	So far we motivated on physical grounds (the properties of the effective Hamiltonian) why maximizing RSMI generally provides a guiding principle for constructing a real-space RG  procedure. We then investigated on the example of the 1D Ising system the properites of such a scheme in a typical situation, when the RSMI maximization problem is additionally constrained by the number and type of degrees of freedom the system is coarse-grained into. In particular, we gave physical intuitions which justify the solution RSMI converges to in the 1D case, i.e. decimation. This is to be contrasted with the situation in 2D, when the decimation procedure is known to immediately generate long-range and many-spin interactions and can be shown not to posses a nontrivial fixed-point at all \cite{vanEnter1993}. For the square-lattice Ising model in two dimensions the majority rule transformation is preferable: numerical evidence, at least, points to the existence of a fixed point~\cite{PhysRevB.30.3866}. Remarkably, the RSMI solution in 2D converges (numerically) towards a majority-rule block transformation (for 2-by-2 blocks) \cite{Koch-Janusz2018}. In this section we provide an information-theory based explanation of these observations. In doing so we also elucidate and quantify the non-trivial influence on the RG scheme of the constraints imposed by the properties (type and number) of the new coarse-grained variables, for the general case. Finally, we exemplify our findings using simple and intuitive toy models.

	To this end let us revisit the inequality Eq.(\ref{eq:MIbound_visibles}).  We refine it by explicitly introducing the random variables $\V_{\Lambda}$, which the hidden degrees of freedom $h_i \in \H$ couple to in a RG scheme parametrized by $\Lambda = \{\lambda_{ij}\}$. For instance, in the RBM-parametrization discussed previously, while generically $\H$ depends on the full $\V$, the coarse-graining defined by the conditional probability $P_{\Lambda}(\H|\V)$ only makes each $h_i\in\H$ dependent on the combination:
	\begin{equation}\label{eq:vlambda}
	\V_{\Lambda_i} = \frac{1}{||\Lambda_i||}\sum_j \lambda_{ij}v_j.
	\end{equation}
	Note that the overall normalization in the definition is not important, but only the relative stregths of $\lambda_{ij}$ which define the linear combination of degrees of freedom in the block. The following now holds:
	\begin{equation}\label{eq:infeqsnew1}
	I_{\Lambda}(\H : \E) \leq I(\V_{\Lambda} : \E) \leq I(\V : \E),
	\end{equation}
	that is: the information about the environment carried by the particular chosen variables $\V_{\Lambda}$ is potentially smaller that the overall information about the environment contained in the block $I(\V : \E)$. Still less of the information may ultimetely be encoded in the degrees of freedom $\H$. 
	
	Where do the inequalities Eq.(\ref{eq:infeqsnew1}) originate from? Formally this is because we have a Markov chain:
	\begin{equation}\label{eq:markov1}
	\E \rightarrow \V \rightarrow \V_\Lambda \rightarrow \H,
	\end{equation}
	but the more pertinent question is what can make those inequalities sharp. The second one is rather trivial: if we only decide to keep a few (one, as is often the case) variables $\V_{\Lambda_i}$, then their entropy may be simply too small to even store the full information $I(\V : \E)$. Still, for the same entropy, there may be choices of $\Lambda$ which result in bigger or smaller $I(\V_{\Lambda_i} : \E)$. Crucially though, $I(\V_{\Lambda_i} : \E)$ \emph{does not} depend on the nature of $h_i\in\H$ (i.e. on whether $h_i$ is a binary variable or not, for instance). It only characterizes how good the particular set of physical degrees of freedom $\V_\Lambda$ is at describing fluctuations in the environment $\E$. 
	
	Whether this information can be efficiently encoded in $\H$ is a different question entirely. The answer, and the origin of the first inequality Eq.(\ref{eq:infeqsnew1}), is revealed by:
	\begin{equation}\label{eq:chainruleMILambda}
	I_{\Lambda}(\H:\E) = I(\V_\Lambda :\E) - I(\V_{\Lambda}:\E | \H),
	\end{equation}
	where $I(\V_{\Lambda}:\E | \H)$ is the conditional mutual information and we have used the chain rule and the Markov property Eq.(\ref{eq:markov1}). Since $I(\V_{\Lambda} : \E)$ is independent of $\H$ in the sense described above, $I(\V_{\Lambda}:\E | \H)$ quantifies the failure of the encoding into $\H$ due to the properties of the $\H$ itself (conditional mutual information being always non-negative). We have thus managed to identify the contributions to $I_{\Lambda}(\H:\E)$ resulting from coupling to a certain choice of physical modes in $\V$,  and to isolate them from the losses incurred due to impossibility of encoding this information perfectly in a particular type of $\H$.
	
	The conditional probabilty distribution $I(\V_{\Lambda}:\E | \H)$  can be thought of as describing the mismatch of the probability spaces of the random variables $\H$ and $\V_\Lambda$, it tells us how much information is \emph{still} shared between $\V_\Lambda$ and $\E$ after $\V_\Lambda$ has been restricted to only values compatible with a given outcome of $\H$. For example, in the 1D Ising case we examined previously, the majority rule defines $\V_\Lambda = v_1 + v_2$, for which the set of possible outcomes is equivalent to $\{-1,0,1\}$. The entropy of $\V_\Lambda$ is bounded by and possibly equal to $\log_2(3)$. Since the system is $\mathbb{Z}_2$ symmetric, then unless $\mbox{Prob}[\V_\Lambda=0] = 0$, this cannot be faithfully encoded into \emph{any} probability distribution of a single binary variable $\H$. 
	Below we construct simple toy models to provide  more examples and intuitions for the somewhat abstract notions we introduced here.
	
	First, let us stress though, that the RSMI prescription maximizes $I_{\Lambda}(\H:\E)$ as a whole, and that, for a type of $\H$ fixed at the outset, the procedure \emph{cannot} be split into maximization of $I(\V_\Lambda :\E)$ followed by a linear coupling of $\H$ to the $\V_\Lambda$ found. Such a naive greedy approach does not necessarily lead to an optimal solution -- the toy models below provide an explicit counterexample. The RSMI-based solution of $P_{\Lambda}(\H|\V)$ thus converges to the optimal trade-off between finding the best modes in $\V$ to describe $\E$, and finding those, whose description can be faithfully written in $\H$ of a given type.
	
		 	\begin{figure}
		 		\centering
		 		\includegraphics{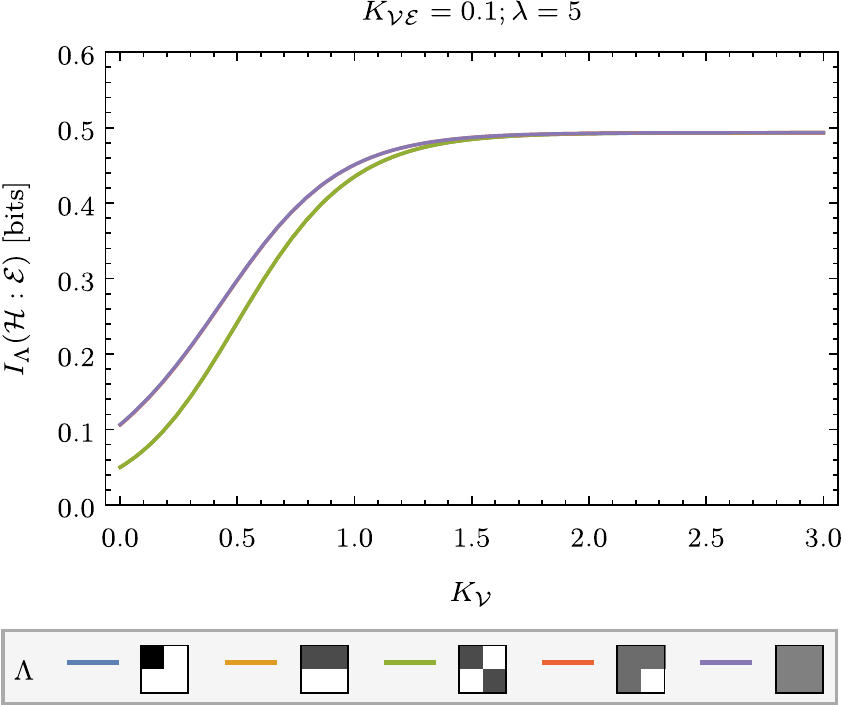}
		 		\caption{
		 			Mutual information $I_{\Lambda}(\H:\E)$ in the toy model of a 2D system Eq.(\ref{eq:toymodel_2D_Ham}), as a function
		 			of the coupling $K_\V$ between the visibles. The coupling pattern to the visibles in different RG rules $\Lambda$ is shown schematically in the (physically relevant) large $||\Lambda||^2$ limit. The majority rule (and interestingly, also coupling to three spins, depicted by a red line, coinciding with the purple one) consistently retains more information than decimation, or any coupling to two spins (blue and yellow, coinciding with green). Again, the distinction vanishes for large $K_\V$ when all visible spins are bound into a single one.
		 		}
		 		\label{fig:toymodel2d2}
		 	\end{figure}	
	
	To illustrate the above considerations we construct minimal toy models. In 1D this consists of four coupled Ising spins: $v_1,v_2$ in the block $\V$, and $e_1,e_2$ representing the left- and right- environment (in 1D the environment is not simply connected), with the Hamiltonian:
	\begin{equation}
	\K = K_{\V\E} (e_1v_1+v_2e_2) + K_\V v_1v_2,
	\label{eq:toymodel_1D_Ham}
	\end{equation}
	where, as before the coupling constants contain a factor of $\beta = 1/k_BT$.
	 The two spins in $\V$ are coupled to a single hidden spin $\H$ using an RBM-ansatz Eq.(\ref{eq:RG_rule_RBM_ansatz}) and the random variable $\V_\Lambda$ is defined as in Eq.(\ref{eq:vlambda}). In Fig.(\ref{fig:toymodel1d}) the results of the calculation of the mutual informations $I_{\Lambda}(\H:\E)$ and $I(\V_\Lambda :\E)$ for decimation and  the majority rule are shown. In the regime of strong coupling to the environment $K_{\V\E}$ [see Fig.(\ref{fig:toymodel1d}a)], for small $K_\V$ both visible spins are nearly independent and almost copy the state of the left- and right-environments, respectively. Consequently, $\V_\Lambda$ for the majority rule carries almost $\log_2(3)$ bits of information about the environment while  $\V_\Lambda$ for decimation, being a binary variable, at most one bit. However, when $I_{\Lambda}(\H:\E)$ is examined it becomes aparent that for decimation it is exactly equal to $I(\V_\Lambda :\E)$, while for majority rule it is significantly lower, so much so, that overall decimation is better across the whole parameter regime! The difference between the solid and dashed curves in Fig.(\ref{fig:toymodel1d}a) is precisely the mismatch of Eq.(\ref{eq:chainruleMILambda}), and the above provides a counterexample to a greedy maximization of $I(\V_\Lambda :\E)$ instead of $I_{\Lambda}(\H:\E)$, which was mentioned previously. In the large $K_\V$ limit both spins in $\V$ become bound into an effective single binary variable and the distinction between the two rules vanishes. In Fig.(\ref{fig:toymodel1d}b) we show the same in the regime when the spins in $\V$ are only weakly coupled to the environment (or the temperature is high). Again, decimation perfectly encodes information $I(\V_\Lambda :\E)$ into $\H$ and is overall better.
	 
	 Let us contrast this with the situation in higher (in particular: two) dimensions, when the environment is simply connected. Based on the discussion above, we may anticipate that the optimal solution could be different, and that majority rule may instead be preferable. This is because, on the one hand, for the same coupling strength to the environment and the same linear dimensions $L_{\V}=2$ of the block, the ratio of $I(\V_\Lambda :\E)$ for the majority rule to the one for decimation increases with increasing dimension (consequence of all visible spins interacting with \emph{the same} environment). On the other hand the mismatch $I(\V_{\Lambda}:\E | \H)$ for majority rule decreases, compared to 1D, since the probability of $\V_\Lambda = \sum_i v_i$ being zero is smaller. This fact is due both to dimensional considerations, as well as (again) the environment being simply connected, the importance of which, even in 1D, we illustrate in Appendix \ref{appendix:toymodel}.
	 
	 We verify those expectations using a simple toy model of the 2D setting: the environment is represented by a single random variable $E$ with a large number of states, to which all the spins in $\V$ couple. These states should be thought of intuitively as fluctuations of some large environment at wave-lenghts longer than the size of the coarse-graining cell. The Hamiltonian is:
	 	\begin{IEEEeqnarray}{rCl}
	 		\K &=& K_{\V\E} E (v_1+v_2+v_3+v_4)\IEEEnonumber\\
	 		&&\quad+\,K_\V (v_1v_2+v_1v_3+v_2v_4+v_3v_4).
	 		\label{eq:toymodel_2D_Ham}
	 	\end{IEEEeqnarray}
	 	As before, the spins in block $\V$ are coupled to a single hidden spin $\H$ with an RBM-ansatz parametrized by $\Lambda$.

	In Fig.(\ref{fig:toymodel2d2}) the mutual information $I_{\Lambda}(\H:\E)$ is computed for the model Eq.(\ref{eq:toymodel_2D_Ham}) for different course-graining rules given by $\Lambda$.  Indeed, the decimation is now inferior to the majority rule across the full parameter range. 
	This is also consistent with the known properties of decimation and majority rule for the 2D Ising model, and suggests their information-theoretic origin.

	
	\section{Disordered systems}\label{sec:disorder}
			\begin{figure}[t]
				\centering
				\includegraphics[width = 0.49\textwidth]{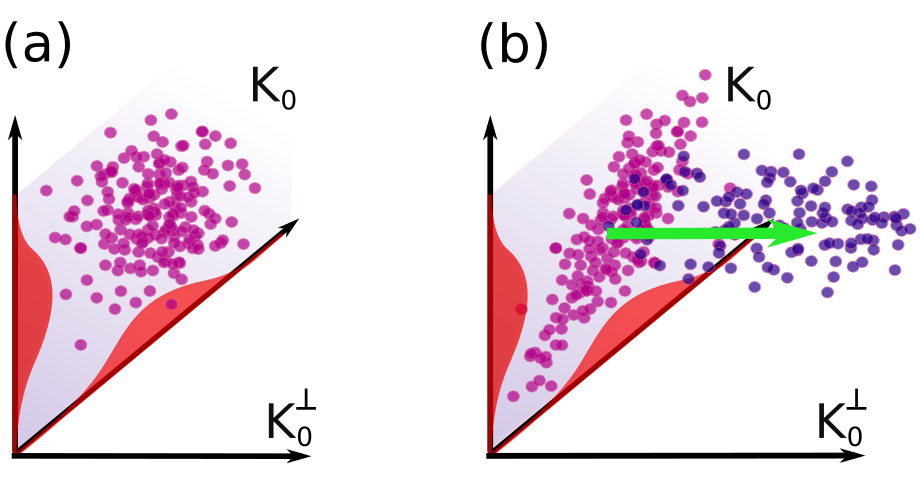}
				\caption{Generic behaviour of disorder probability under RG. Every point is a disorder realization (Hamiltonian). The shaded plane $K_0$ denotes the subspace of 2-body nearest-neighbour couplings; $K_0^\perp$ is the complement space of all other couplings. (a) Initially a \emph{factorized} distribution is usually assumed, schematically shown as a product of independent Gaussians. (b) After the RG step the distribution can develop correlations, depicted for simplicity in the $K_0$ plane, and additional couplings can be generated, resulting in probability mass leaking out of $K_0$. The former can be quantified by KL-divergence or distance correlation, the latter by the shift of the centre-of-mass of the distribution, depicted with a green arrow (see also Fig.~\ref{fig:disorderCorCM})
				}
				\label{fig:densityplot_disorder1}
			\end{figure}
	
	While investigations of clean higher-dimensional models (to which RSMI can be applied without any restriction), such as e.g.~the 3D Ising model, are still relevant, of much more interest are disordered systems. We show that RSMI naturally generalizes to this case, and that the information theoretic approach provides important insights, particularly concerning disorder correlations.
	
	The proper object of study in the disordered setting is not the individual Hamiltonian ${\Ham}$, but rather the disordered Hamiltonian \emph{distribution} $\Prob(\Ham)$ \cite{PhysRevLett.36.415,PhysRevLett.36.1508,doi:10.1063/1.324982}, which equivalently can be thought of as a distribution over the vector space spanned by all the possible coupling constants $\{K_{i_1,i_2,\ldots,i_M}\}$. Denoting the (potentially infinite dimensional) vector of couplings by $\bf{K}$, the RG transformation induces a mapping:
	\begin{equation}\label{eq:distr1}
	\Prob(\bf{K}) \rightarrow \Prob'(\bf{K}),
	\end{equation}
	generating RG-flows of $\Prob$ with fixed point distributions $\Prob^*$. The formalism subsumes the clean case: any fixed Hamiltonian is a trivial delta-like distribution with all probability mass concentrated in one point.	

				\begin{figure*}[th]
					\centering
					\includegraphics[width = 0.99\textwidth]{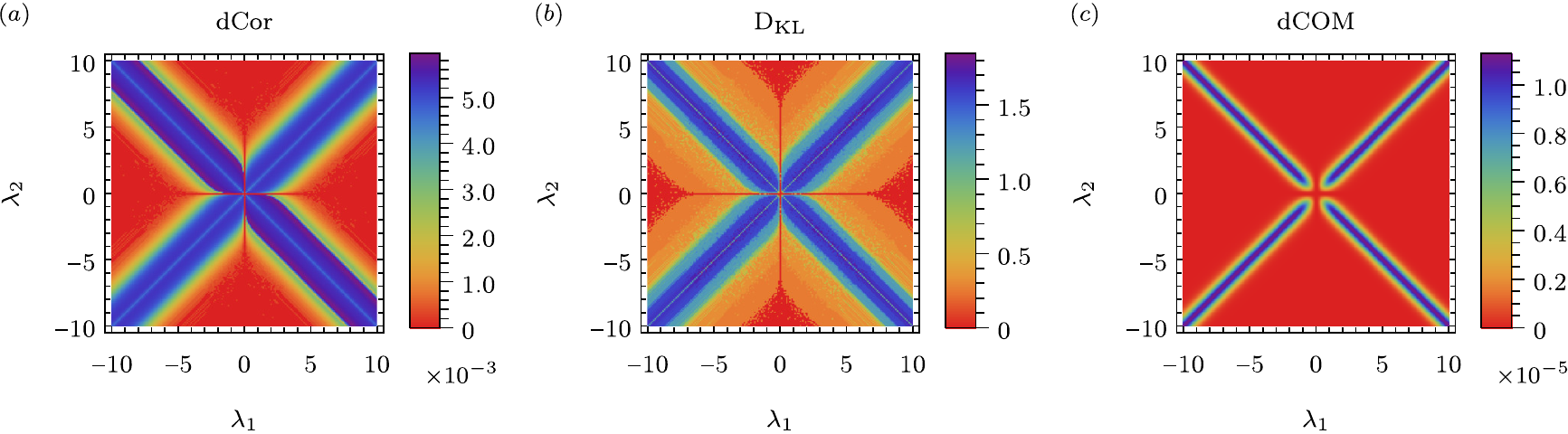}
					\caption{Properties of the renormalized disorder distribution as a function of the coarse-graining rule for the case of random dilute Ising chain. The RG rules are parametrized by $\lambda_1$, $\lambda_2$, as before. (a) The distance correlation (see the main text) $\rm{dCor}$ between renormalized distributions of two neighbouring NN couplings. The couplings are uncorrelated if and only if $\rm{dCor}$ vanishes -- compare with Fig.~\ref{fig:densityplotlv2}a. (b) An alternative measure of correlations is the Kullback-Leibler divergence $\rm{D}_{\rm{KL}}$, computed between the product of marginal distributions of neighbouring couplings and their joint distribution
						(c) The $K_0^\perp$ centre of mass of the disorder distribution $\rm{dCOM}$ (see the main text and Fig.~\ref{fig:densityplot_disorder1}). Note that all those quantities vanish when MI is maximized.
					}
					\label{fig:disorderCorCM}
				\end{figure*}
	
	Let us examine the mapping Eq.~(\ref{eq:distr1}). The probabilistic framework of Sec.~\ref{sec:RSMI} can also be used in this case.
	For any fixed disorder realization $\bf{K}$ the RG transformation (conditional probability distribution) is applied to the Gibbsian probability measure defined by the Hamiltonian $\Ham(\bf{K})$, and the new effective Hamiltonian $\Ham(\bf{K'})$ is implicitly defined exactly as in Eq. (\ref{eq:RG_IT_Ham}). The new coupling constants are in this way the functions of the old ones: $\bf{K'} = \bf{K'}(\bf{K})$, and can be recovered by solving the inverse problem. 
	Their distribution is obtained by integrating over $\Prob(\bf{K})$:
		\begin{equation}\label{eq:distr2}
		\Prob'(\bf{K'}) = \int \delta(\bf{K'}-\bf{K'}(\bf{K}))\Prob(\mathbf{K}) d\mathbf{K}
		\end{equation}
	Equation (\ref{eq:distr2}) appears trivial, but of course all the complexity of the problem is concealed in the functional dependence of $\bf{K'}$ on $\bf{K}$. The distribution $\Prob$ is usually assumed to be factorized into a product of independent distributions, over, say, bond strengths \cite{PhysRevLett.36.415,PhysRevLett.36.1508,doi:10.1063/1.324982,PhysRevB.23.3421,PhysRevE.64.066107,Angelini3328}. The flow of the distribution is then analyzed either analytically, or numerically, in terms of a variable characterizing the strength of disorder, i.e. the variance of the individual distribution factor in $\Prob$ \cite{Angelini3328}, by forcing a factorized parametrization at each stage. It is clear, however, that even if this (often unrealistic, since one can expect disorder in nearby areas to be correlated \cite{crystallographyDisorder}) assumption holds initially, the renormalized distribution need not necessarily necessarily obey it, except in special cases. Generically, coarse-graining the system introduces correlations in $\Prob$.  Additionally, as in the clean case, higher order and longer range couplings are generated, in effect shifting the disorder distribution away from the hyperplane defined by only nearest-neighbour couplings. Both effects, depicted in Fig.~\ref{fig:densityplot_disorder1}, increase the complexity of distribution and render the problem of computing and analyzing RG flows for disordered systems very challenging.

	The real-space RG transformations applied to disordered systems are either similar to those used in the clean case, i.e.~various decimation/Migdal-Kadanoff prescriptions, or based on the strong disorder RG \cite{PhysRevLett.43.1434,PhysRevB.22.1305}. We focus on block transformations, which have the advantage of maintaining a regular topology in higher dimensions \cite{sdrg-misc} (though the arguments below apply also when coarse-graining cells are chosen in a sequential, greedy fashion). The very same questions as in the translation-invariant setting need to be answered: is there a more fundamental reason - beyond a simple algebraic coincidence - why certain RG transformations should work better in particular cases? Is there a \emph{constructive} way to find the best such transformation within a certain class, for a given physical system?
	
	Our results suggest, that the answer to both questions is affirmative: beyond controlling the range of the interactions, RSMI maximization also suppresses generation of correlations in the renormalized distribution $\Prob$. As in Secs. \ref{sec:effHam} and \ref{sec:Ising}, we first prove that factorizability properties of $\Prob$ under the full-information capture assumption are stable to local changes in disorder, at least in (quasi)-1D systems, and \emph{imply suppression of disorder correlations}.
   Subsequently, we study the effect of arbitrary RG rules on the renormalized disorder distribution using a model system where the optimal solution is known, and the distribution can be computed (perturbatively) for arbitrary transformation.  
   The following counterpart to the Proposition 1 holds true (proof in Appendix \ref{appendix:proofRange}):
		\begin{proposition} Consider a disordered 1D system, with a factorizable (product) disorder distribution over, without loss of generality, nearest neighbour couplings. The choice $\Lambda^*$ of the optimal coarse graining of a block $\X_0$, satisfying $I(\X_0':\E_0) = I(\X_0:\E_0)$, and thus the factorization property of Proposition 1, are stable to local changes in disorder, provided those do not affect directly the block or the buffer, i.e.~are fully confined to the environment.
		\end{proposition}
		
		Proposition 2 has two important consequences: (i) as seen from Eq.~\ref{eq:facto_appendix_th_tilde} of the proof, in the explicit factorization of the conditional probability distribution \emph{of the coarse-grained degrees of freedom} in the left and right environments (cf. Proposition 1), changes to the disorder realization in one of the coarse-grained environments do not affect the distribution of degrees of freedom in the other. This implies:
		\begin{corollary}
		The probability distribution of degrees of freedom in $\E_R(\X'_0)$ being completely insensitive to the choice of \emph{disorder realization} in $\E_L(\X'_0)$, there cannot exist any correlations \emph{in the renormalized disorder distribution} $\mathcal{P}$ between the regions $\E_L(\X'_0)$ and $\E_R(\X'_0)$ (i.e. no such correlations across $\X'_0$ are generated by the optimal coarse-graining). 
		\end{corollary}
		Since this holds for every block, we conclude that with the above assumptions disorder correlations remain suppressed under coarse-graining. Note that this can be generalized to higher dimensions similarly to Sec.~\ref{sec:effHam}. (ii) Proposition~2 also implies that for the purpose of finding the optimal course-graining of a block, which in general should depend on the disorder realization (as is also the case in strong-disorder RG), only the disorder configuration in the local neighbourhood of the block can be considered. Though this is strictly true under the full information capture assumption, it provides motivation for constructing adaptive coarse-grainings in more complicated systems, with the RG rule optimized for the \emph{local} disorder realization (or, more practically, the equivalence classes thereof).

		We turn to a solvable model system to empirically demonstrate decay of disorder correlations as a function of MI, especially when the stringent requirements of Proposition 2 are not satisfied. In the random Ising chain \cite{PhysRevLett.36.1508} the nearest-neighbour couplings are independent random variables $K_{i,i+1}$ distributed according to a probability $\Prob(K)$. For a generic $\Prob(K)$ the recursive RG Eq.~\ref{eq:distr2} is intractable, but for the special case when it is Bernoulli distributed:
		\begin{equation}\label{eq:Bernoulli}
		\Prob(K) = p\delta(K-k_0) + (1-p)\delta(K-k_1),
		\end{equation}
		the decimation transformation allows to solve Eq.~\ref{eq:distr2} analytically, since the factorizability is then preserved exactly along the flow. The model exhibists much richer phenomenology than the clean case: for $k_0=-k_1$, in addition to the the usual (unstable) ferromagnetic and (stable) paramagnetic fixed points, the spin-glass fixed point is reached for any $0<p<1$ if starting exactly at $T=0$. For $k_1=0$,  i.e. in the random \emph{dilute} Ising chain, additional Griffiths singularities  appear in the limit $k_0\rightarrow \infty$ and $h_0\rightarrow 0$, where $h_0\rightarrow 0$ are the on-site (uniform) magnetic fields~\cite{PhysRevB.10.4665,PhysRevLett.36.1508}. This is associated with existence of rare, but arbitrarily large coupled clusters of spins~\cite{PhysRevLett.23.17}.		
		
		 We focus on the random dilute Ising chain, but allow instead any arbitrary RG transformation (without loss of generality, for a block of two sites) parametrized by $\Lambda = (\lambda_1,\lambda_2)$, as in Sec.~\ref{sec:Ising}. For a finite periodic system the renormalized couplings can be computed perturbatively, using the cumulant expansion, for any quenched disorder realization and any $\Lambda$. For illustration, we consider a system of 16 spins and \emph{all} possible disorder realizations. For each realization we compute the Hamiltonian after the RG step, for arbitrary $\Lambda$, by summing up to ninth order in the cumulants, obtaining the full renormalized disorder distribution $\Prob^1_\Lambda(\bf{K})$, where $\bf{K}$ is a vector of all possible couplings between the block spins.

		To quantify the generated disorder correlations in $\Prob^1_\Lambda(\bf{K})$, we examine the joint probability distribution of two neighbouring NN couplings $\Prob^1_\Lambda(K_{i,i+1},K_{i+1,i+2})$, obtained by marginalization, as a function of $\Lambda$ (it was chosen since those correlations develop the fastest). We use two statistical measures of dependence for this distribution: the \emph{distance correlation} $\rm{dCor}$ \cite{10.2307/25464608} and the information-theoretic \emph{Kullback-Leibler divergence} $\rm{D}_{\rm{KL}}$ \cite{kullback1951}. Both are sensitive also to nonlinear correlations, and share the essential property that two random variables are statistically independent if and only if $\rm{dCor}=\rm{D}_{\rm{KL}}=0$, though distance correlation is generally better suited for continuous variables.
		In Fig.~\ref{fig:disorderCorCM}a we plot $\rm{dCor}(K_{i,i+1},K_{i+1,i+2})$ as a function of $\Lambda$, while in Fig.~\ref{fig:disorderCorCM}b $\rm{D}_{\rm{KL}}(\Prob^1_\Lambda(K_{i,i+1},K_{i+1,i+2}) || \Prob^1_\Lambda(K_{i,i+1})\Prob^1_\Lambda(K_{i+1,i+2}))$. Both measures coincide: the disorder distributions at neighbouring bonds are the more independent, the more RSMI is retained by the coarse-graining rule, as seen by comparing with Fig.~\ref{fig:densityplotlv2}a (which is valid, up to rescaling, for every quenched disorder realization in the model). The couplings are statistically independent, rendering the renormalized disorder distribution factorizable, precisely where RSMI is maximized, i.e.~for decimation. This empirically establishes decay of correlations.
		
		We also investigate another measure of complexity i.e.~how non-nearest-neighbour terms are generated as a function of $\Lambda$. Denote by $\rm{K}_0$ the subspace of 2-body NN couplings, and by $\rm{K}_0^\perp$ the orthogonal space of all other couplings (see Fig.~\ref{fig:densityplot_disorder1}). For any (renormalized) disorder realization $\bf{K}$ let $\bf{K}^\perp$ be its restriction to $K^\perp$, obtained by truncation of all couplings in $\rm{K}_0$. One measure of the shift of the renormalized disorder distribution away from $K_0$ is the $K_0^\perp$-center-of-mass $\rm{dCOM}$:
		\begin{equation}\label{eq:COM}
		\rm{dCOM} = \lVert \sum_{\bf{K}} \Prob(\bf{K})\bf{K}^\perp \rVert,
		\end{equation}
		where $\lVert \cdot \rVert$ is the Euclidean norm. It is shown in Fig.~\ref{fig:disorderCorCM}c, and exhibits the same qualitative behaviour, as a function of $\Lambda$, as the correlation measures, i.e.~it vanishes as a function of increasing RSMI.
		
		We thus observe that empirically RSMI maximization suppresses generation of both spurious correlations and of higher order and long-range couplings in the renormalized disorder distribution, which is also supported by Proposition 2 (under the appropriate assumptions).

		
	\section{Conclusions and Outlook}\label{sec:conclusions}
	
	We investigated information-theoretic properties of real-space RG procedures, and particularly of one
	based on variational maximization of real-space mutual information (RSMI) \cite{Koch-Janusz2018}, both for clean and disordered systems. We demonstrated suppression of longer range interactions in the renormalized Hamiltonian as a function of RSMI retained: formally, proving this statement under explicitly stated assumptions, and empirically, by considering arbitrary coarse-graininings in the solvable example of the Ising chain. For the case of disordered systems, again using formal proofs and the example of dilute random Ising chain, we showed that in addition to longer-range/higher-order terms, also correlations in the renormalized disorder distribution are suppressed. We also examined the effect of constraints on the type of coarse grained variables on the RG procedure.

	Our results provide a formal underpinning for the \emph{physical} intuition behind the RSMI maximization: the effective long-wavelength description of the system is simple in terms of degrees of freedom which carry the most information about its large scale behaviour. While the notion of ``simplicity" may be ambiguous -- despite the clear practical consequences of its absence  -- as there exist multiple measures of Hamiltonian complexity, the long-range information and its retention can be defined rigourously, similarly to Information Bottleneck approach of compression theory \cite{infbottle1}.
	Different measures of complexity we computed for both clean (range, amount of n-body interactions) and disordered systems (correlations in the renormalized disorder distribution measured using KL-divergance and distance correlation) are \emph{all} suppressed as more RSMI is preserved by the coarse-graining. This strongly indicates that the model-independent RSMI coarse-graining, optimal by construction from the point of view of compression theory, is also optimal physically, resulting in operationally desirable properties (a tractable Hamiltonian). We thus established direct and quantifiable connections between the information theoretic properties of RG transformation and the actual physical properties of the renormalized Hamiltonian and the disorder distribution.

Beyond conceptual significance, the results can be useful practically, inspiring new numerical approaches to RG for disordered/complex systems, as briefly discussed in Sec.~\ref{sec:disorder}. This is especially interesting given progress in machine learning, and numerical techniques for MI estimation \cite{pmlr-v80-belghazi18a,Hernandez:2019qgp}, and the inverse problem \cite{Lokhove1700791}.

	 A number of distinct further research directions are possible. On the formal part of the spectrum, a mathematically rigourous investigation of the probability measure defined by the RSMI coarse-graining, in the spirit of Refs.~\cite{vanEnter1993,Kennedy1997,Kennedy2010}, is desirable. Conceptually, an interesting question is whether the type and number of coarse-grained variables can also be variationally optimized, as opposed to being chosen at the outset, as is usually the case. This would have the interpretation of ``discovering" whether the best variables to describe a system, originally given in terms of, say, Ising spins, are the same, or rather some emergent degrees of freedom are preferable (see also Refs. \cite{Ronhovde2011,Ronhovde2012,doi:10.1021/acs.jpclett.8b00733}). More practically, the results invite the application of the RSMI method to the study of disordered systems, both using synthetic, as well as experimental data.
	Finally, a quantum version of the procedure is an open question. The Information Bottleneck has recently been extended to the case of compression of quantum data \cite{2017arXiv170402903S}: in this setting the conditional probability of classical systems is replaced by a quantum channel. It would be interesting to explore how the physics of the system manifests itself in properties of these optimal channels, and to compare it with energy based approaches \cite{PhysRevD.54.4131,PhysRevB.65.104508,PhysRevLett.90.147204,PhysRevLett.93.187205} and the recently introduced Gilt-TNR method \cite{PhysRevB.97.045111}.

	\section*{Acknowledgements}
	We thank Prof. Gianni Blatter for his insightful comments. S.D.H. and M.K-J. gratefully acknowledge the support of Swiss National Science Foundation (SNSF).


\begin{thebibliography}{87}
\bibitem [{opt()}]{optimal-misc}%
\BibitemOpen
\href@noop {} {}\bibinfo {note} {For a discussion of other works defining RG
	procedures which are ``optimal" in a specific context we refer to
	Appendix~\ref{appendix:comparison}.}\BibitemShut {Stop}%
\bibitem [{bot()}]{bottleneck-misc}%
\BibitemOpen
\href@noop {} {}\bibinfo {note} {The RSMI scheme can be understood as a
	realization of the famous Information Bottleneck compression of relevant
	information \cite{infbottle1}, under the constraint that the compressed
	variable is of predetermined type and size -- a single block spin, for
	instance -- as opposed to introducing continuous Lagrange multiplier for its
	entropy (as originally). In this setting the signal $\V$ is being (lossy)
	compressed to $\H$ and the ``relevancy" variable determining features of
	signal to be preserved is the environment $\E$.}\BibitemShut {Stop}%
\bibitem [{sdr()}]{sdrg-misc}%
\BibitemOpen
\href@noop {} {}\bibinfo {note} {A block method inspired by strong disorder
	RG has also been proposed \cite{MonthusGarel_strong}}\BibitemShut {NoStop}%
\end{thebibliography}
	%

	\appendix
	
	
	\section{Mutual information} \label{appendix:MI}
	The mutual information  Eq.(\ref{eq:MI_def}) can equivalently be defined by:
	\begin{equation}
	I_\Lambda(\H:\E) = H(\H)-H(\H|\E),
	\label{eq:appendix:MI_chain_rule}
	\end{equation}
	where:
	\begin{IEEEeqnarray}{rCl}
		H(\H) &=& -\sum_\H P_\Lambda(\H)\log(P_\Lambda(\H)),\\
		H(\H|\E) &=& -\sum_{\H,\E}P_\Lambda(\H,\E)
		\log\left(\frac{P_\Lambda(\H,\E)}{P(\E)}\right),
	\end{IEEEeqnarray}
	are the Shannon entropy and conditional entropy, respectively. It is a symmetric quantity. Positivity of mutual information and of the conditional entropy, together with
	the bound on entropy, immediately imply the following inequalities:
	\begin{equation}
	0 \leq I_\Lambda(\H:\E) \leq H(\H),
	\label{eq:appendix:MIbound_entropy}
	\end{equation}
	where $H(\H)$ is the entropy of $\H$. The mutual information $I_\Lambda(\H:\E)$ is also bounded by the mutual information of the visibles and the environment:
	\begin{equation}
	I_\Lambda(\H:\E) \leq I(\V:\E),
	\end{equation}
	which is obvious, since the hidden degrees of freedom only couple to the environment via the visibles.
	
	Throughout the text we also use the notion of conditional mutual information, which, for any random variables $\E$,$\H$,$\V$, generically can be defined via the so-called chain rule:
	\begin{equation}
	I(\E:\H,\V) = I(\E:\H) + I(\E:\V|\H).
	\label{eq:appendix:MI_chain_rule_cond}
	\end{equation}

		
	\section{RSMI does not increase the range of interactions and maintains factorizability of distributions}\label{appendix:proofRange}
	Here we give the details of the argument in Sec. \ref{sec:effHam}. We work directly in $D$ dimensions and we spell out explicitly the additional (reasonable) assumptions required, compared to the 1D case. 
	
	Consider a generic finite-ranged Hamiltonian with degrees of freedom $\X$ in $D$ dimensions. For concreteness let us assume a hypercubic lattice. We partition $\X$ into hypercubic coarse-graining blocks $\V_j$ large enough, so that only nearest-neighbour blocks interact. For the purpose of this argument we arrange the blocks into parallel $(D-1)$-dimensional hyperplanes index by $l$ so that $\X=\cup_\ell \X_\ell$ with $\X_\ell=\cup_{j\in J_\ell}\V_j$. Thus, in terms of the hyperplanes we end up with a quasi-one-dimensional structure. Let us choose an arbitrary hyperplane $\X_0$, denote its immediate neighbours $\X_{\pm1}$ as the buffer $\B$, and the union of the remaining hyperplanes $\X_{\ell<-1}$ and $\X_{\ell>1}$ as the environment $\E_0(\X_0)$, or, in more detail, as left- and right-environment $\E_{L/R}(\X_0)$, respectively.
	
	Assume now that the coarse-grained variables $\H_j$ for the blocks $\V_j$ in $\X_0$ are constructed in such a way that $I(\X_0':\E_0) = I(\X_0:\E_0)$, where $\X_0'=\cup_{j\in J_0}\H_j$.
	This is the full information capture condition for the hyperplane, generalizing the condition for the single block in 1D (note though, that we still optimize variables $\H_j$ for each block, and not some new collective hidden variables for the  entire hyperplanes). Strictly speaking, this requires an additional assumption (compared to 1D) that it is equivalent to assuming $I(\H_j:\E_0) = I(\V_j:\E_0)$ separately for each individual block in the hyperplane $\X_0$. This seems reasonable for a short-ranged Hamiltonian, at least in the isotropic case. Under those assumptions we show the probability measure on the coarse-grained variables $P(\X')$ obeys a $D$ dimenional analogue of factorization property of Proposition 1 in Sec. \ref{sec:effHam}: 
	\begin{equation}\label{eq:facto_appendix}
		P(\X'_{j \leq-2},\X'_{j \geq2}| \X'_0) = P(\X'_{j\leq-2} |\X'_0)\cdot P(\X'_{j\geq2}| \X'_0).
	\end{equation}
	To prove this we begin with a crucial separability lemma, which is a technical condition enabling Proposition 1:
	\begin{lemma} Let $\E_L$, $\E_R$ be the left/right environments of $\X_0$ and let $I(\X_0':\E_0) = I(\X_0:\E_0)$. Then the following factorization property with respect to the coarse-grained variable $\X'_0$ holds:
		$P(\E_L,\E_R|\X'_0) = P(\E_L|\X'_0)P(\E_R|\X'_0)$
	\end{lemma}
	\begin{proof}
	To show that, first note that from the full information capture assumption it follows that:
	\begin{IEEEeqnarray}{rCl}
		I(\E_0:\X_0|\X'_0) &=& I(\E_0:\X_0,\X'_0)-I(\E_0:\X'_0)\IEEEnonumber\\
		&=& I(\E_0:\X_0)-I(\E_0:\X'_0) = 0,
	\end{IEEEeqnarray}
	where the first equality is the chain rule for mutual information and the second is due to the fact that the coarse-grained variables $\X'_0$ are a function of $\X_0$ only. Vanishing of this mutual information is equivalent to the (conditional) probability distribution factorizing and therefore:
	\begin{equation}\label{eq:appendix:env_factorization0}
		P(\E_0,\X_0|\X'_0) = P(\E_0|\X'_0)P(\X_0|\X'_0).
	\end{equation}
	Furthermore, the locality of the interactions assumption implies that $I(\E_L:\E_R|\X_0) = 0$ and thus:
	\begin{IEEEeqnarray}{rCl}
		P(\E_L,\E_R,\X_0|\X'_0) &=& P(\E_L,\E_R|\X_0,\X'_0)P(\X_0|\X'_0)
		\IEEEnonumber\\
		&=& P(\E_L,\E_R|\X_0)P(\X_0|\X'_0)\IEEEnonumber\\
		&=& P(\E_L|\X_0)P(\E_R|\X_0)P(\X_0|\X'_0).\IEEEnonumber\\
		\label{eq:appendix:env_factorization}
	\end{IEEEeqnarray}
	Comparing Eqs. (\ref{eq:appendix:env_factorization0}) and (\ref{eq:appendix:env_factorization}) we find that:
	\begin{equation}
		P(\E_0|\X'_0)P(\X_0|\X'_0) = P(\E_L|\X_0)P(\E_R|\X_0)P(\X_0|\X'_0).
	\end{equation}
	For a given $\X'_0$ let us denote the set of $\X_0$ such that $P(\X_0|\X'_0)\neq 0$ by $\{\X_0(\X'_0)\}$.
	For all such ``compatible" $\X_0\in\{\X_0(\H_0)\}$ we can divide by $P(\X_0|\X'_0)$ and obtain:
	\begin{equation}
		P(\E_0|\X'_0) = P(\E_L|\X_0)P(\E_R|\X_0).
		\label{eq:appendix:env_fact_indep}
	\end{equation}
	Crucially, the left hand side does not depend on $\X_0$, and so long as $\X_0\in\{\X_0(\X'_0)\}$ the equality holds and the conditional probability factorizes independently of particular $\X_0$. In fact the factorization holds generally and the case $P(\X_0|\X'_0) = 0$ is not a problem:
	\begin{IEEEeqnarray}{rCl}
		P(\E_0|\X'_0) &=& \sum_{\X_0} P(\E_L,\E_R,\X_0|\X'_0)\IEEEnonumber\\
		&=& \sum_{\X_0} P(\E_L|\X_0)P(\E_R|\X_0)P(\X_0|\X'_0)\IEEEnonumber\\
		&=& \sum_{\X_0\in\{\X_0(\X'_0)\}} P(\E_L|\X_0)P(\E_R|\X_0)P(\X_0|\X'_0)\IEEEnonumber\\
		&=& P(\E_L|\X_0(\X'_0))P(\E_R|\X_0(\X'_0)),
		\label{eq:appendix:env_fact_indep2}
	\end{IEEEeqnarray}
	where we used Eq. (\ref{eq:appendix:env_factorization}) in the second equality, explicitly removed vanishing (by virtue of $P(\X_0|\X'_0) = 0$) terms in the sum in the third, and used Eq. (\ref{eq:appendix:env_fact_indep}) to take the $\X_0$-independent product from under the restricted summation in the third.  We thus constructed an explicit factorization of $P(\E_L,\E_R|\X'_0)$ in Eq. (\ref{eq:appendix:env_fact_indep2}), which implies:
	\begin{equation}
		I(\E_L:\E_R|\X'_0) = 0
	\end{equation}
	and hence we can simply write:
	\begin{eqnarray}
		P(\E_L,\E_R|\X'_0) = P(\E_L|\X'_0)P(\E_R|\X'_0).
		\label{eq:appendix:env_fact_hidden}
	\end{eqnarray}
	
	\end{proof}
	The Lemma has a very nice physical interpretation, which provides a useful intuition of more general validity. It states that with a short range Hamiltonian an area of finite width has to mediate \emph{all} correlations between its neighbourhoods, and if the information that area has about them is accurately retained in a new variable, no correlation can exist between the neighbourhoods which hasn't got the new variable as an intermediary. Note we have not relied on translation invariance at all. This will be useful in deriving corresponding statement for disordered systems, which explains the results in Sec. \ref{sec:disorder}. More immediately, it is the key element in showing Eq. (\ref{eq:facto_appendix}), giving the $D$ dimensional version of Proposition 1:
	
	\addtocounter{proposition}{-2}
	\begin{proposition}
	Let $I(\X_0':\E_0) = I(\X_0:\E_0)$.  Then the probability measure on the coarse-grained variables $P(\X')$ obeys the factorization property: 
	\begin{equation}\label{eq:facto_appendix_th}
	P(\X'_{j \leq-2},\X'_{j \geq2}| \X'_0) = P(\X'_{j\leq-2} |\X'_0)\cdot P(\X'_{j\geq2}| \X'_0).
	\end{equation}
	where in the conditional probabilties the buffer (i.e. the neighbours $\X'_{\pm 1} $ of $\X'_0$) has been integrated out. In other words, for fixed $\X'_0$ the probabilities of its left and right environments $\E_{L/R}(\X'_0)$ are independent of each other. 
	\end{proposition}
	\begin{proof}
	Consider the coarse-grained probability measure defined by Eqs.(\ref{eq:RG_IT}) and (\ref{eq:rg_rule}):
	\begin{equation}
		P(\X') = \sum_\X P(\X)\prod_j P(\H_j|\V_j).
	\end{equation}
	Denoting the product of the block conditional probability distributions in the hyperplanes by $\prod_\ell P(\X'_\ell|\X_\ell)$ and integrating out $\X'_{\pm 1}$ we have:
	\begin{equation}\label{eq:appendix:splitprob0}
		P(\X'_{j \leq -2},\X'_0,\X'_{j \geq2})
		= \sum_{\X_{\abs{\ell}\neq1}} P(\{\X_\ell\}_{\abs{\ell}\neq 1})
		\prod_{\abs{\ell}\neq 1} P(\X'_\ell|\X_\ell).
	\end{equation}
	Using the definition of conditional probability and the fact that $\X'_0$ only directly depends on $\X_0$ we have:
	\begin{IEEEeqnarray}{rCl}
		P(\{\X_\ell\}_{\abs{\ell}\neq 1})P(\X'_0|\X_0)
		&\equiv& P(\E_0,\X_0,\X'_0)\IEEEnonumber\\
		&=& P(\E_0,\X_0|\X'_0)P(\X'_0),
	\end{IEEEeqnarray}
	which allows us to write:
	\begin{widetext}
	\begin{IEEEeqnarray}{rCl}\label{eq:appendix:fact_calculation1}
	P(\X'_{j \leq -2},\X'_0,\X'_{j \geq2}) &=& \sum_{\X_{\abs{\ell}\neq 1}} P(\E_0,\X_0|\X'_0)P(\X'_0) \prod_{\abs{\ell}\neq 0,1} P(\X'_\ell|\X_\ell) =  \\ \label{eq:appendix:fact_calculation2}
	&=&\sum_{\X_{\abs{\ell}\neq 1}} P(\E_L,\E_R|\X'_0)P(\X_0|\X'_0)P(\X'_0)\prod_{\abs{\ell}\neq 0,1} P(\X'_\ell|\X_\ell) =  \\ \label{eq:appendix:fact_calculation3}
	&=&\sum_{\X_{\abs{\ell}\neq 1}} P(\E_L|\X'_0)P(\E_R|\X'_0)P(\X_0|\X'_0)P(\X'_0)\prod_{\abs{\ell}\neq 0,1} P(\X'_\ell|\X_\ell) =  \\ \label{eq:appendix:fact_calculation4}
	&=&\sum_{\X_{\abs{\ell}\neq 0,1}} P(\E_L|\X'_0)P(\E_R|\X'_0)P(\X'_0)\prod_{\abs{\ell}\neq 0,1} P(\X'_\ell|\X_\ell) = \\ \label{eq:appendix:fact_calculation5}
	&=& P(\X'_0)\left(\sum_{\X_{\ell\leq-2}} P(\E_L|\X'_0)\prod_{\ell\leq-2} P(\X'_\ell|\X_\ell) \right) \left(\sum_{\X_{\ell\geq2}} P(\E_R|\X'_0)\prod_{\ell\geq2} P(\X'_\ell|\X_\ell) \right) = \\ \label{eq:appendix:fact_calculation6}
	&=& P(\X'_0)\cdot P(\X'_{j\leq-2} |\X'_0)\cdot P(\X'_{j\geq2}| \X'_0),
	\end{IEEEeqnarray}
	\end{widetext}
	where to obtain Eq.~(\ref{eq:appendix:fact_calculation2}) we conditioned on $\X_0$ and used the full information capture assumption to write $P(\E_L,\E_R|\X_0,\X'_0) = P(\E_L,\E_R|\X'_0)$, to obtain Eq.~(\ref{eq:appendix:fact_calculation3}) we used the factorization Eq.~(\ref{eq:appendix:env_fact_hidden}) proved in the Lemma, to obtain Eq.~(\ref{eq:appendix:fact_calculation4}) we performed the summation over $\X_0$, to obtain Eq.~(\ref{eq:appendix:fact_calculation5}) we rearranged the sums taking expressions independent of summation variables out of them, and in the last line we used Bayes' law. Dividing both sides  by $P(\X'_0)$ we obtain Eq.~(\ref{eq:facto_appendix_th})
	\end{proof}
	Proposition 1 shows that for a fixed $\X'_0$ the probability $P(\X'_{j \leq -2},\X'_0,\X'_{j \geq2})$ factorizes into a product over left and right environments. As described in the main text, together with the arbitrariness of the choice of the hyperplane this implies (barring a pathological fine-tuned scenario in which integration over $\X_{\pm1}$ \emph{exactly} cancels all pre-existing NNN couplings) that the effective Hamiltonian in terms of new variables is still nearest-neighbour (in all directions).
	
	Furthermore, under the same assumptions of finite-ranged Hamiltonian, we can also derive an important result about the properties of the renormalized disorder distribution. Assume without loss of generality that the blocks are chosen sufficiently large to render interactions nearest-neighbour with respect to the blocks. Then:
	\addtocounter{proposition}{0}
			\begin{proposition} Consider a disordered 1D system, with a factorizable (product) disorder distribution over, without loss of generality, nearest neighbour couplings. The choice $\Lambda^*$ of the optimal coarse graining of a block $\X_0$, satisfying $I(\X_0':\E_0) = I(\X_0:\E_0)$, and thus the factorization property of Proposition 1, are stable to local changes in disorder, provided those do not affect directly the block or the buffer, i.e.~are fully confined to the environment.
			\end{proposition}
		\begin{proof}
For a fixed quenched disorder realization, denote the probability distribution of the degrees of freedom under this Hamiltonian by $P(\X)$. Let $P_{\Lambda^*}(\X'_0|\X_0)$ be the optimal coarse-graining for the block $\X_0$, determined by $\Lambda^*$, saturating mutual information, and consequently ensuring the factorization property of the Lemma is obeyed. Consider now a localized change to the disorder realization, affecting only terms acting entirely within an area $\X_D \subset \E_L$, resulting in a modified probability distribution $\tilde{P}(\X)$. One can then show that the factorization property still holds, \emph{with the very same choice of} $P_{\Lambda^*}(\X'_0|\X_0)$.

To this end denote by $\tilde{K}(\X_D)$ the local terms in the reduced Hamiltonian affected by the disorder change, and by $K(\X_D)$ the original ones (the change to the Hamiltonian $K$ is also localized since it is NN in the blocks, and the change to the disorder is confined to $\X_D$).
Then:
\begin{widetext}
	\begin{IEEEeqnarray}{rCl}\nonumber 
		\tilde{P}(\E_L,\E_R,\X_0)P_{\Lambda^*}(\X'_0|\X_0) &=& \frac{e^{\tilde{K}(\E_L,\E_R,\X_0)}}{\tilde{Z}}  P_{\Lambda^*}(\X'_0|\X_0)=  
		\frac{e^{K(\E_L,\E_R,\X_0)}}{Z}\frac{Z}{\tilde{Z}} \frac{e^{\tilde{K}(\X_D)}}{e^{K(\X_D)}}   P_{\Lambda^*}(\X'_0|\X_0)= \\ \nonumber
		&=& \frac{Z}{\tilde{Z}} \frac{e^{\tilde{K}(\X_D)}}{e^{K(\X_D)}} P(\E_L,\E_R,\X_0)P_{\Lambda^*}(\X'_0|\X_0) = \frac{Z}{\tilde{Z}} \frac{e^{\tilde{K}(\X_D)}}{e^{K(\X_D)}} P(\E_L,\E_R,\X_0,\X'_0) = \\ \nonumber
		&=& \frac{Z}{\tilde{Z}} \frac{e^{\tilde{K}(\X_D)}}{e^{K(\X_D)}} P(\E_L,\E_R,\X_0|\X'_0) P(\X'_0) = \\ \nonumber
		&=& \frac{Z}{\tilde{Z}} \frac{e^{\tilde{K}(\X_D)}}{e^{K(\X_D)}} P(\E_L|\X'_0) P(\E_R|\X'_0)P(\X_0|\X'_0)P(\X'_0) = \\ \label{eq:appendix:propo2_1}
		&=& \tilde{P}(\E_L|\X'_0) P(\E_R|\X'_0)P(\X_0|\X'_0)P(\X'_0)
	\end{IEEEeqnarray}
\end{widetext}	
with $Z$ and $\tilde{Z}$ are the original and modified partition functions. In the penultimate line we used the Lemma for the initial distribution $P(\E_L,\E_R,\X_0)$, and in the last we absorbed all additional factors, which are local, into the definition of $\tilde{P}(\E_L | \X'_0)$. Dividing both sides by  $P(\X'_0)$ and marginalizing over $\X_0$ we arrive at:
\begin{equation}\label{eq:stable_fac}
\tilde{P}(\E_L,\E_R|X'_0) = \tilde{P}(\E_L|\X'_0) P(\E_R|\X'_0),
\end{equation}
where the right factor is \emph{as for the original distribution}.

Since the change to the disorder (other than being confined to $\E_L$) was completely arbitrary, equation \ref{eq:stable_fac} shows any such localized changes do not affect the choice of optimal coarse graining of the block. They thus do not break the factorization property and, in particular, do not affect the other environment: note that in Eq. \ref{eq:stable_fac} we still have the \emph{original} $P(\E_R | \X'_0)$! Consequently the Proposition~1 immediately holds for both the original and modified disorder realization, with the same coarse-graining $P_{\Lambda^*}(\X'_0|\X_0)$, and the same probability distribution of the renormalized right environement $\E_R(\X'_0)$:
	\begin{equation}\label{eq:facto_appendix_th_tilde}
	\tilde{P}(\X'_{j \leq-2},\X'_{j \geq2}| \X'_0) = \tilde{P}(\X'_{j\leq-2} |\X'_0)\cdot P(\X'_{j\geq2}| \X'_0).
	\end{equation}
	Hence, the coarse-graining is stable.
			\end{proof}
			
Proposition 2 has an important consequence for the renormalized \emph{disorder distribution} $\mathcal{P}$:
 the probability distribution of degrees of freedom in $\E_R(\X'_0)$ being completely insensitive to the choice of \emph{disorder realization} in $\E_L(\X'_0)$, we conclude that there cannot exist any correlations \emph{in the disorder distribution} between in the regions $\E_L(\X'_0)$ and $\E_R(\X'_0)$ (i.e. no such correlations across $\X'_0$ are generated by the optimal coarse-graining).

	
	\section{The effective Hamiltonian}	
	
	\subsection{The cumulant expansion}\label{sec:app_cumulantexp}
	
	Consider a generic Hamiltonian $\K[\X]$. We split it into two parts~\cite{Niemeyer1974}:
	\begin{equation}
	\K[\X] = \K_\intra[\X] + \K_\inter[\X],
	\label{eq:app_Ham_decomp_intra_inter}
	\end{equation}
	where $\K_\intra$, contains \emph{intra}-block terms, i.e. those which only couple  spins within a single block, and $\K_\inter$ contains \emph{inter}-block terms,
	i.e. those that couple spins from different blocks.	 Such a decomposition simplifies the calculations significantly.
	For translationally invariant systems the intra-block terms are all of the same form:
	\begin{equation}
	\K_\intra[\X] = \sum_{j=1}^n\K_\block[\V_j].
	\label{eq:app_Kintra}
	\end{equation}
	
	Using the decomposition Eqs.(\ref{eq:app_Ham_decomp_intra_inter}) and (\ref{eq:app_Kintra}) the definition of the renormalized Hamiltonian in Eq.(\ref{eq:RG_IT_Ham}) can be rewritten as an intra-block average of the inter-block part of the Hamiltonian:
	\begin{IEEEeqnarray}{rCl}
		e^{\K'[\X']}
		&=& Z_\intra\sum_\X e^{\K_\inter[\X]}
		\prod_{j=1}^n\underbrace{\frac{e^{\K_\block[\V_j]}}{Z_\block}P_\Lambda(\H_j|\V_j)}_{=:P_{\restblock}(\H_j,\V_j)}\IEEEnonumber\\
		&=& Z_\intra\sum_\X e^{\K_\inter[\X]}
		\prod_{j=1}^nP_{\restblock}(\V_j|\H_j)P_{\restblock}(\H_j)\IEEEnonumber\\
		&=& Z_\intra\underbrace{\prod_{j=1}^nP_{\restblock}(\H_j)}_{=:P_\restintra(\X')}
		\sum_\X e^{\K_\inter[\X]}\underbrace{\prod_{j=1}^nP_{\restblock}(\V_j|\H_j)}_{=:P_{\restintra}(\X|\X')}\IEEEnonumber\\
		&=& Z_\intra P_\restintra(\X')\avg{e^{\K_\inter[\X]}}_{\restintra}[\X'],
		\label{eq:app_RGeq_for_cumulants}
	\end{IEEEeqnarray}
	where the average $\avg{\cdot}_\restintra$ is over $P_\restintra(\X|\X')$
	as a probability distribution in $\X$ and thus introduces a
	dependence on the new spin variables $\X'$.
	We indicate this dependence by square brackets $[.]$ after the average.
	
	Equation \eqref{eq:app_RGeq_for_cumulants} lends itself to a cumulant expansion:
	\begin{equation}
	\avg{e^{\K_\inter[\X]}}_\restintra[\X']
	=e^{\sum_{k=0}^\infty\frac{1}{k!}C_k[\X']}
	\label{eq:app_cumulant_expansion}
	\end{equation}
	with the standard expressions for the cumulants in terms of moments.  The first few of which are given by:
	\begin{IEEEeqnarray}{rCl}
		\IEEEyesnumber\label{eq:app_cumulants_def}
		C_1 &=& \avg{\K_\inter}_\restintra,\IEEEyessubnumber\\
		C_2 &=& \avg{\K_\inter^2}_\restintra - \avg{\K_\inter}_\restintra^2,
		\IEEEyessubnumber\\
		C_3 &=& \avg{\K_\inter^3}_\restintra - 3\avg{\K_\inter^2}_\restintra\avg{\K_\inter}_\restintra + 2\avg{\K_\inter}_\restintra^3,
		\IEEEyessubnumber
	\end{IEEEeqnarray}
	where for brevity we did not indicate the dependence on $\X'$.
	The powers of $\K_\inter$ inside the averages induce
	couplings between multiple blocks, and naturally lead to new coupling terms
	in the effective Hamiltonian.
	
	The cumulant expansion Eq.(\ref{eq:app_cumulant_expansion}) allows to determine the new Hamiltonian by taking the logarithm of Eq.(\ref{eq:app_RGeq_for_cumulants}):
	\begin{equation}
	\K'[\X'] = \log(Z_\intra P_\restintra(\X')) + \sum_{k=0}^\infty\frac{1}{k!}C_k[\X'].
	\label{eq:app_ren_Ham}
	\end{equation}
	The renormalized coupling constants are not	apparent in Eq.(\ref{eq:app_ren_Ham}).
	In order to identify them we introduce the following canonical form of the Hamiltonian:
	\begin{equation}
	\K'[\X'] = K_0' + \sum_{\{\alpha_\ell\}_{\ell=1}^n}
	K'_{\alpha_1,\alpha_2,\dotsc,\alpha_n}
	\left(\sum_{j=1}^n\prod_{\ell=1}^n(x_{j+\ell}')^{\alpha_\ell}\right),
	\label{eq:app_Ham_canonical}
	\end{equation}
	with $\alpha_1= 1$ and $\alpha_\ell\in\{0,1\}$ for all $\ell>1$.
	Here, addition of the indices is to be understood modulo $n$ (i.e. with periodic boundary conditions).
	Note that arbitrary orders $k$ of the cumulant expansion $C_k$ contribute to each coupling constant $K'_{\alpha_1,\alpha_2,\dotsc,\alpha_n}$.

	\subsection{Factorization of quenched averages}
	\label{appendix:avgs_factorize}
	Factorization of the conditional probability distribution results in the factorization of expectations $\avg{\mathcal{O}[\X]}_\restintra$ for any operator,
	which is a product of operators $o_j$ acting on separate blocks, i.e. $\mathcal{O}[\X]=\prod_{j=1}^no_j[\V_j]$:
	\begin{IEEEeqnarray}{rCl}
		\avg{\mathcal{O}[\X]}_\restintra[\X']
		&=& \prod_{j=1}^n
		\sum_{\V_j}o_j[\V_j]P_\restblock(\V_j|\H_j)\IEEEnonumber\\
		&=& \prod_{j=1}^n\avg{o_j[\V]}_\restblock[\H_j],
	\end{IEEEeqnarray}
	where the probability over which we average is:
	\begin{equation}
		P_\restblock(\V|\H)
		=\frac{e^{\K_\block[\V]}}{Z_\block P_\restblock(\H)}P_\Lambda(\H|\V).
		\label{eq:appendix:cumulant_dist}
	\end{equation}
	In particular, the factorization holds for the operators $\K_\inter^k$, which appear in
	the expressions for the cumulants.


	\subsection{Parametrization of the RG rule using RBM ansatz}
	\label{appendix:RBMansatz}
	
	The conditional probability distribution $P_\Lambda(\H|\V)$ is parametrized using a Restricted Boltzmann Machine (RBM)  ansatz~\cite{Ackley1985,Hinton1986,Salakhutdinov2012}.
	The RBMs belong to a family of energy-based models, whose main purpose is to efficiently approximate probability distributions, and, more generally, they are an example of a growing class of machine learning techniques recently employed in a statistical physics or condensed matter setting~\cite{Wang2016,Torlai2016,NIPS2016_6211,Carrasquilla2017,Nieuwenburg2017,Carleo602,Koch-Janusz2018,2018arXiv180202840L,PhysRevB.97.045153,PhysRevB.97.085104,2018arXiv180205267F}.
	
	In the RBM ansatz the joint probability of the visible and hidden degrees of freedom is approximated by a Boltzmann distribution:
	\begin{equation}\label{eq:botlzmann_param1}
	P(\V,\H) = \frac{1}{Z}e^{-E_\Lambda (\V,\H)},
	\end{equation}
	with a quadratic energy energy function:
	\begin{equation}
		E_\Lambda(\V,\H)
		= -\sum_{i,j} \lambda_{i}^{j}v_ih_j - \sum_i \alpha_i v_i
		- \sum_j \beta_j h_j,
	\end{equation}
	where $v_i\in\V$, $h_j\in\H$ and $\Lambda$ collectively denotes the set of parameters $\{\lambda_i^j\}_{i,j}$, $\{\alpha_i\}_i$ and $\{\beta_j\}_j$, which are to be variationally optimized
	so that $P_{\Lambda}(\V,\H)$ they define is as close as possible to the target distribution $P(\V,\H)$. Note that the energy function only couples the visible to the hidden degrees of freedom, and includes no couplings \emph{within} the visible or the hidden sets. This pecularity (which the word ``restricted" in RBM refers to) is crucial to the existence of fast algorithms~\cite{hinton2} for training and sampling from the trained distribution $P_{\Lambda}(\V,\H)$.
	
	The conditional probability is then given by:
	\begin{equation}
	P_\Lambda(\H|\V)
	= \frac{e^{-E_\Lambda(\V,\H)}}{\sum_{\H}e^{-E_\Lambda(\V,\H)}}.
	\label{eq:appendix:RG_rule_genRBM_ansatz0}
	\end{equation}
	It is easy to see that the parameters $\{\alpha_i\}_i$ drop out in $P_\Lambda(\H|\V)$. Additionally, because of the Ising $\mathbb{Z}_2$ symmetry the bias (magnetic field) term for $h_j$ is not allowed: $\beta_i=0$ for all $i$. Due to the absence of interactions between hiddens, the expression factorizes and the summation over $\H$ is trivial. In the case of a 1D system and a single hidden spin $\H=\{h\}$ the conditional probability is then given explicitly by:
	\begin{equation}
		P_\Lambda(\H|\V)
		= \frac{1}{1+e^{-2h\sum_{i=1}^{L_\V}\lambda_i v_i}},
		\label{eq:appendix:RG_rule_genRBM_ansatz}
	\end{equation}
	with $\Lambda=\{\lambda_i\}_i$. The choice of the parameters \emph{defines} the RG rule. It is intuitively clear that while one could, in principle, consider any choice of $\Lambda$, the physically meaningful choices would correspond to the limit $||\Lambda||^2 \rightarrow\infty$, i.e. when the value of $h$ actually strongly depends on ${v}$. In that limit Eq.(\ref{eq:appendix:RG_rule_genRBM_ansatz}) becomes a Heaviside function. This is also what happens in practice during the RSMI training (see Supplemental Materials in Ref.~\cite{Koch-Janusz2018}).
	
	Thus the virtue of the RBM ansatz is twofold: first, it provides an efficient tool from the algorithmic perspective of RSMI implementation, and second, it also provides a well-behaved, differentiable analytical ansatz, which we use to explicitly calculate the quantities of interest. We emphasize though, that conceptually the RBM ansatz is not essential to the RSMI approach. Any other parametrization of $P_{\Lambda}(\H,\V)$ can also be used, at the expense of having to devise efficient algorithms to fix parameters of this new ansatz.

	
	\section{The 1D Ising model}\label{app:Ising}
	For the 1D Ising model, Eq.(\ref{eq:Ham_Ising}) and a single hidden spin we define:
	\begin{IEEEeqnarray}{rCl}
		\IEEEyesnumber\label{eq:appendix:partitions}
		\V_j &=& \{x_{(j-1)L_\V+1},x_{(j-1)L_\V+2},\dotsc,x_{jL_\V}\},
		\IEEEyessubnumber\\
		\H_j &=& \{h_j\}.\IEEEyessubnumber
	\end{IEEEeqnarray}
	The Hamiltonian decomposition Eq.(\ref{eq:Ham_decomp_intra_inter}) gives:
	\begin{IEEEeqnarray}{rCl}
		\K_\block[\V] &=& K\sum_{i=1}^{L_\V-1}v_iv_{i+1},
		\label{eq:appendix:Ising_Kblock}\\
		\K_\inter[\X] &=& K\sum_{j=1}^{n}x_{jL_\V}x_{jL_\V+1}
		\label{eq:appendix:Ising_Kinter}
	\end{IEEEeqnarray}
	with the partition functions $Z_\intra = \prod_{j=1}^n Z_\block$, where
	$Z_\block = \sum_\V e^{\K_\block[\V]}$.
	
	The 1D Ising model with nearest neighbor interactions can be solved exactly using the method of transfer matrices.
	To this end define the transfer matrix $T$ with components:
	$\matrixel{x_1}{T}{x_2}:=e^{Kx_1x_2}$.
	The matrix elements of arbitrary integer powers of $T$ can be computed by diagonalization:
	\begin{equation}
	\matrixel{x_1}{T^m}{x_2}
	= \frac{1}{2}\left(2\cosh(K)\right)^m\left(1+x_1x_2\tanh(K)^m\right).
	\label{eq:appendix:transfermatrixelement}
	\end{equation}
	
	\subsection{Exact decimation}
	\label{appendix:decimation}
	For the purpose of numerical comparison with the RSMI solution we perform one step of the \emph{exact} decimation RG transformation Eq.(\ref{eq:decimation_rule}).
	Following Eq.(\ref{eq:RG_IT_Ham}):
	\begin{IEEEeqnarray}{rCl}
		e^{\K'[\X']}
		&=& \sum_\X \prod_{j=1}^n P_\Lambda(\H_j|\V_j)
		e^{K\sum_{i=1}^N x_ix_{i+1}}\IEEEnonumber\\
		&=& \sum_\X\prod_{j=1}^nP_\Lambda(x_j'|\{x_{2j-1},x_{2j}\})
		\IEEEnonumber\\
		&&\quad\times\matrixel{x_{2j-1}}{T}{x_{2j}}\matrixel{x_{2j}}{T}{x_{2(j+1)-1}}
		\IEEEnonumber\\
		&=& \prod_{j=1}^n\matrixel{x_j'}{T^2}{x_{j+1}'},
	\end{IEEEeqnarray}
	because for every block $j$ the delta-like conditional probability $P_\Lambda(x_i'|\{x_{2j-1},x_{2j}\})$ strictly
	enforces $x_j'=x_{2j-1}$ and does not involve $x_{2j}$. Thus, $x_{2j}$ can simply be integrated out.
	The above has, up to a multiplicative constant $e^{c'}$, the same form as $e^{\K[\X]}$ with a new coupling constant $K'$, such that we can set
	$e^{c'}T'=T^2$. From that we obtain:
	\begin{IEEEeqnarray}{rCl}
		c' &=& \frac{1}{2}\log(4\cosh(2K)),\IEEEyessubnumber\\
		K' &=& \frac{1}{2}\log(\cosh(2K)),\IEEEyessubnumber
	\end{IEEEeqnarray}
	such that the renormalized Hamiltonian is:
	\begin{equation}
	\K'[\X'] = \frac{n}{2}\log(4\cosh(2K)) + K'\sum_{i=1}^n x_i'x_{i+1}'.
	\end{equation}
	
	\subsection{The effective Hamiltonian}

	\label{appendix:Zb_proof}
	Here we compute the effective block parameters Eq.(\ref{eq:effbparams}) of the 1D Ising model for general block size $L_\V$.
	
	Using Eq.(\ref{eq:appendix:transfermatrixelement}), the partition function of intra-block contribution to Hamiltonian is given by:
	\begin{IEEEeqnarray}{rCl}
		Z_\block &=&
		\sum_\V\prod_{i=1}^{L_\V-1}\matrixel{v_i}{T}{v_{i+1}} = \sum_{v_1,v_{L_\V}}\matrixel{v_1}{T^{L_\V-1}}{v_{L_\V}}\IEEEnonumber\\
		&=& 2(2\cosh(K))^{L_\V-1}.
	\end{IEEEeqnarray}
The expectations of powers of inter-block couplings $\K_1[\X]^k$ appearing in the cumulant expansion can be written as a sum of
products of operators acting on single blocks (see Appendix \ref{appendix:avgs_factorize}). We have:
\begin{IEEEeqnarray}{rCl}
	\K_\inter[\X]^k &=& \left(K\sum_{j=1}^{n}x_{j\cdot L_\V}x_{j\cdot L_\V+1}\right)^k\\
	&=& K^k \sum_{\sum_{j=1}^nk_j=k}\frac{k!}{\prod_{j=1}^nk_j!}
	\underbrace{\prod_{j=1}^n\left(x_{j\cdot L_\V}x_{j\cdot L_\V+1}\right)^{k_j}}_{
		=:\mathcal{O}}.\IEEEnonumber
\end{IEEEeqnarray}
We now consider one term in the above sum and rearrange the factors
according to blocks:
\begin{IEEEeqnarray}{rCl}
	\mathcal{O}	&=& \prod_{j=1}^n
	\underbrace{x_{(j-1)L_\V+1}^{k_{j-1}}x_{jL_\V}^{k_j}}_{=:o_j}.
\end{IEEEeqnarray}
Depending on the values of $k_{j-1}$ and $k_j$, the block-operator $o_j$
is one of the following three operators: $x_{(j-1)L_\V+1}$, $x_{jL_\V}$ or $x_{(j-1)L_\V+1}x_{jL_\V}$.
Hence, the average $\avg{\mathcal{O}}_\restblock$ factorizes into:
\begin{IEEEeqnarray}{c}
	\avg{x_{(j-1)L_\V+1}}_\restblock[\H_j],\IEEEyessubnumber\\
	\avg{x_{jL_\V}}_\restblock[\H_j],\IEEEyessubnumber\\
	\avg{x_{(j-1)L_\V+1}x_{jL_\V}}_\restblock[\H_j].\IEEEyessubnumber
\end{IEEEeqnarray}
The $\mathbb{Z}_2$ symmetry of the 1D Ising model can be used to extract the dependence of $P_\restblock(h)$ and the above three quantities on the single hidden spin $h$:
	\begin{IEEEeqnarray}{rCl}
		P_\restblock(h) &=& \sum_\V\frac{e^{\K_\block[\V]}}{Z_\block}
		\underbrace{P_\Lambda(h|\V)}_{=P_\Lambda(-h|-\V)}\\
		&=& \sum_\V \frac{e^{\K_\block[\V]}}{Z_\block}P_\Lambda(-h|\V) = P_\restblock(-h), \IEEEnonumber
	\end{IEEEeqnarray}
	where we used the fact that for $\mathbb{Z}_2$-symmetric system the coarse-graining satisfies $P_\Lambda(\H|\V) = P_\Lambda(-\H|-\V)$. Since $P_\restblock(h)$ is normalized we have:
	\begin{eqnarray}
	P_\restblock(h) = \frac{1}{2}.
	\label{eq:PLambdah}
	\end{eqnarray}
	For any operator $\mathcal{O}_p[\V]$ with definite $\V$-parity $p=\pm1$ given by $\mathcal{O}_p[-\V]=p\mathcal{O}_p[\V]$, we find using similar arguments that:
	\begin{equation}
		\avg{\mathcal{O}_p[\V]}_\restblock[h]
		= p\avg{\mathcal{O}_p[\V]}_\restblock[-h],
	\end{equation}
	since $P_\restblock(-\V|-h) = P_\restblock(\V|h)$. Hence $\avg{\mathcal{O}_p[\V]}_\restblock[h]$ also has definite $h$-parity $p$.
	
	Since $h$ \emph{only} assumes values $\pm 1$,  then $p=+1$ implies that the average is actually independent of $h$, while $p=-1$ implies it is linear in $h$.
	Thus:
	\begin{IEEEeqnarray}{rCl}
		\avg{x_{(j-1)L_\V+1}}_\restblock[h] &=&
		\avg{v_1}_\restblock[1] \cdot h,\IEEEyessubnumber \quad \\
		\avg{x_{jL_\V}}_\restblock[h] &=&
		\avg{v_{L_\V}}_\restblock[1] \cdot h,\IEEEyessubnumber \quad \\
		\avg{x_{(j-1)L_\V+1}x_{jL_\V}}_\restblock[h] &=&
		\avg{v_1v_{L_\V}}_\restblock[1]. \quad \IEEEyessubnumber 
	\end{IEEEeqnarray}
	The last expression can actually be explicitly calculated,  independently of the choice of RG rule:
	\begin{IEEEeqnarray}{rCl}\label{eq:appendix:b_calc}
		\avg{v_1v_{L_\V}}_\restblock[1] &=& \left(2\cosh(K)\right)^{-(L_\V-1)}\\
		&&\quad\times
		\sum_{\V}v_1v_{L_\V}e^{\K_\block[\V]}P_\Lambda(1|\V).\IEEEnonumber 
	\end{IEEEeqnarray}
	Since  $v_i=\pm 1$ we also have:
	\begin{IEEEeqnarray}{rCl}
		e^{\K_\block[\V]} &=& \prod_{i=1}^{L_\V-1}e^{v_iv_{i+1}}\\
		&=&
		\cosh(K)^{L_\V-1}\prod_{i=1}^{L_\V-1}\left(1+v_iv_{i+1}\tanh(K)\right).
		\IEEEnonumber
	\end{IEEEeqnarray}
	Every term in the expanded expression is of the form
	$\mathcal{O}[\V]\tanh(K)^m$ for an operator $\mathcal{O}$, which is a
	product of several consecutive pairs $v_iv_{i+1}$. 
	If $\mathcal{O}$ has even $\V$-parity and $v_1v_{L_\V}\mathcal{O}[\V]$ is
	not independent of $\V$, then:
	\begin{IEEEeqnarray}{rCl}
		\sum_{\V}v_1v_{L_\V}\mathcal{O}[\V]P_\Lambda(1|\V) &=&
		\sum_{\V}v_1v_{L_\V}\mathcal{O}[\V]\underbrace{P_\Lambda(-1|-\V)}_{
			=1-P_\Lambda(1|-\V)}\IEEEnonumber\\
		&=& -\sum_{\V}v_1v_{L_\V}\mathcal{O}[\V]P_\Lambda(1|\V)\IEEEnonumber\\
		&=& 0.
	\end{IEEEeqnarray}
	Thus, only $\mathcal{O}$ of odd $\V$-parity and those for which $v_1v_{L_\V}\mathcal{O}[\V]$ is independent of $\V$ can contribute to Eq.(\ref{eq:appendix:b_calc}).
	However, $e^{\K_\block[\V]}$ contains only two such contributions:
	$1$ and $v_1v_{L_\V}\tanh(K)^{L_\V-1}$.
	It follows that:
	\begin{equation}
		\avg{v_1v_{L_\V}}_\restblock[1] = \tanh(K)^{L_\V-1} =: b,
	\end{equation}
	i.e. it is a $\Lambda$-independent constant. The remaining two averages depend on the choice of $\Lambda$, and closed expressions for them are given below for the case of block size $L_\V=2$.
		
	As discussed previously, the cumulants can be expressed in terms of the effective block-parameters  Eqs.(\ref{eq:effbparams}). The actual computations can be done by brute-force summation of all possible terms in Eq.(\ref{eq:K1k}). This, however, is rather impractical for obtaining higher order cumulants. We have instead implemented a simple algorithm based on the combinatorial considerations discussed in the main text. 
	
	\subsection{The case of $L_\V = 2$  blocks: discussion of the numerical results}\label{appendix:lv2plots}
	Specializing to blocks of two visible spins results in:
	\begin{equation}
	\K'[\X'] = \frac{N}{2}\log(2\cosh(K)) + \sum_{n=0}^\infty\frac{1}{n!}C_n(\X')
	\end{equation}
	and the effective block-parameters are found to be:
	\begin{IEEEeqnarray}{rCl}\label{eq:app_blockparams}
		a_1 &=&
		\frac{2\left(\cosh(\lambda_1)\sinh(\lambda_1)+\cosh(\lambda_2)\sinh(\lambda_2)\tanh(K)\right)}{\cosh(2\lambda_1)+\cosh(2\lambda_2)},
		\IEEEnonumber\\
		a_2 &=&
		\frac{2\left(\cosh(\lambda_2)\sinh(\lambda_2)+\cosh(\lambda_1)\sinh(\lambda_1)\tanh(K)\right)}{\cosh(2\lambda_1)+\cosh(2\lambda_2)},
		\IEEEnonumber\\
		b &=& \tanh(K).
	\end{IEEEeqnarray}

	\begin{figure}
		\centering
		\includegraphics{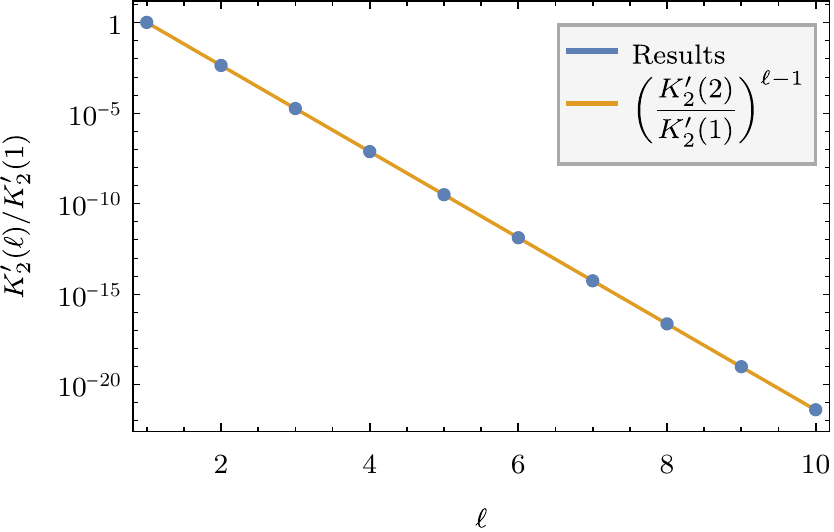}
		\caption{
			Logarithmic plot showing the exponential decay of the two-point
			correlator $K_2'(\ell)$ with distance $\ell$ at $K=0.1$.
			The blue data points represent the results obtained from the
			cumulant expansion of the RSMI-favoured solution up to tenth order, 
			while the yellow line shows the exponential decay with decay length obtained from the first
			two points. For small $K$, where the cumulant expansion is
			expected to be accurate, the two-point correlator decays
			exponentially.
		}
		\label{fig:appendix:k2decay}
	\end{figure}
	
	\begin{figure}
		\centering
		\includegraphics{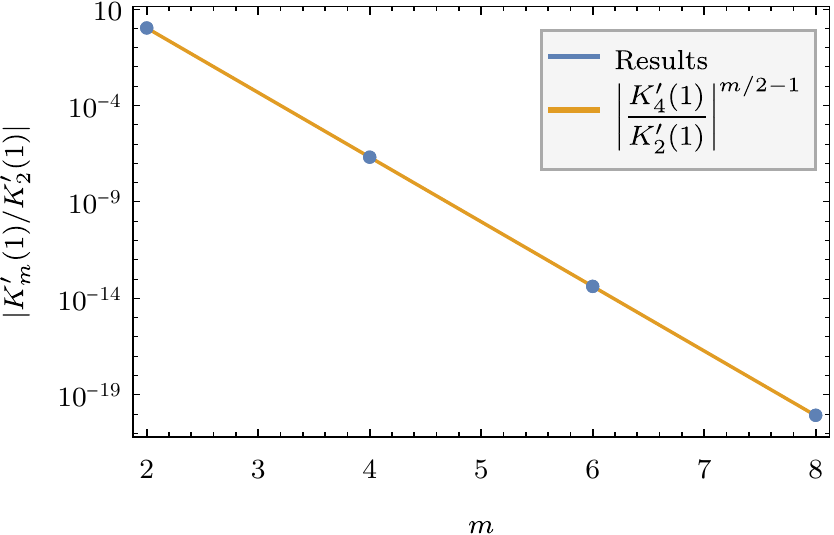}
		\caption{
			Logarithmic plot showing the exponential decay of the nearest neighbor $m$-point correlator $K_m'(1)$ with $m$ at $K=0.1$. The blue data points represent the results obtained from the
			cumulant expansion of the RSMI-favoured solution up to tenth order, while the yellow line shows the exponential decay with decay length obtained from the first two points.
			Only points for even $m$ are present, as $K_m(1)=0$ for odd $m$ due to reasons of symmetry. Again, at small $K$, the two-point correlator decays exponentially.
		}
		\label{fig:appendix:kmdecay}
	\end{figure}
	
	As discussed in the main text, both the two-point correlator as a function of distance between the spins [Fig.(\ref{fig:appendix:k2decay})] and the $m$-point correlator as a function of the number of consecutive spins $m$ [Fig.(\ref{fig:appendix:kmdecay})] decay exponentially for small $K$ for the RSMI-favoured solution (i.e. decimation). This solution, unsurprisingly, is decimation, which can be seen from Figs.(\ref{fig:densityplotlv2}) and (\ref{fig:rangenesslv2}). Additionally in Fig.(\ref{fig:convergence}) we show the convergence to large-$\lambda$ results shown in Fig.(\ref{fig:rangenesslv2}a) with increasing order of cumulant expansion.
	
	\begin{figure}
		\centering
		\includegraphics{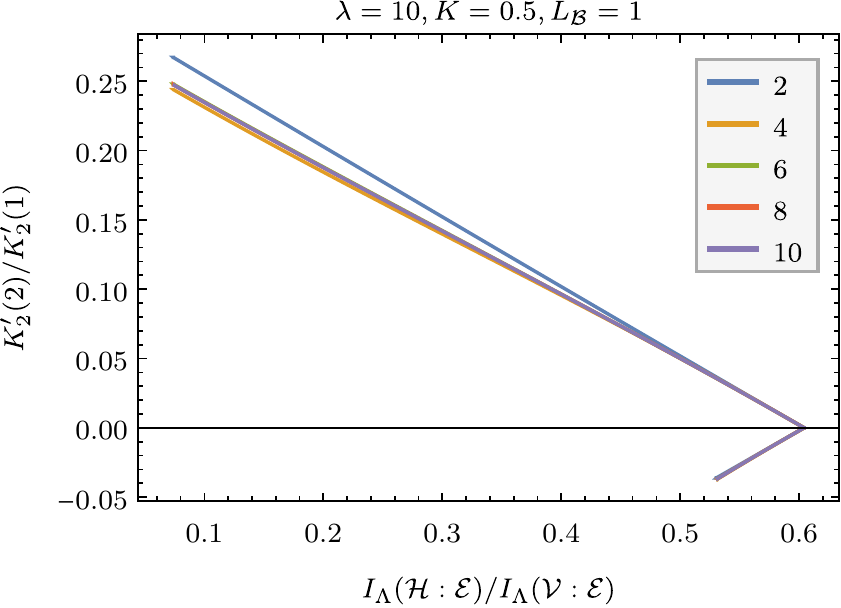}
		\caption{The ratio between next-nearest-neighbour and nearest-neighbour
			coupling constants is plotted versus the mutual information for
			different orders of the cumulant expansion (see inset legend).
			The curves are obtained by parametrizing the RG rule by
			$(\lambda_1,\lambda_2)=\lambda(\cos(\theta),\sin(\theta))$ and
			varying $\theta\in[0,\pi]$.
		}
		\label{fig:convergence}
	\end{figure}
	
	We also comment on the asymmetry (around 0) of the curves in Figs.(\ref{fig:rangenesslv2}a,b). The curves result from traversing the path $\lambda(\cos\theta,\sin\theta)$ in Fig.(\ref{fig:densityplotlv2}), which is \emph{not} fourfold symmetric (instead there are two reflection symmetries with respect to the diagonals). Starting from $\theta=0$ at the peak, the trajectory traces out the lower branch of the curves in Figs.(\ref{fig:rangenesslv2}a,b) reaching the lowest point at $\theta = \pi/4$, before turning around and exactly retracing the trajectory towards the peak at $\theta=\pi/2$. The trajectory then moves on the upper branch reaching the uppermost point at $\theta = 3\pi/4$ and retracing towards peak again at $\theta = \pi/2$. This exact retracing is due to two independent $\mathbb{Z}_2$ symmetries: that of the Ising model and that of the mutual information. Since $\mathbb{Z}_2\times\mathbb{Z}_2$ is not isomorphic to $\mathbb{Z}_4$ we do not have a fourfold symmetry in Fig.(\ref{fig:densityplotlv2}) and consequently we do not have a symmetry around 0 in Figs.(\ref{fig:rangenesslv2}a,b). Physically this is easily  understood: the mutual information in Fig.(\ref{fig:densityplotlv2}a) on the $\lambda_1 = -\lambda_2$ diagonal is lower than on the $\lambda_1 = \lambda_2$ one since for the ferromagnetic Ising model we simulated the neighbouring spins are more likely to be aligned than not. Then for the majority of the spin configurations we have: $\lambda_1v_1 - \lambda_2v_2 = 0$ on the $\lambda_1 = -\lambda_2$ diagonal and hence the coarse-graining rule decides the orientation of the effective spin at random, reducing the mutual information.
	
	We emphasized before that the physically relevant coarse-graining rules are in the limit of large $||\Lambda||^2$. For small values of $||\Lambda||^2$ the coarse-graining rule is essentially independent of the underlying variables $\V$ (or equivalently the rule can be thought of as having a large white noise component). This manifests itself in Fig.(\ref{fig:densityplotlv2}a) by low mutual information in the centre. Nevertheless Figs.(\ref{fig:densityplotlv2}b,c) seem to have some (differently looking) areas of vanishing ``rangeness" and ``m-bodyness" ratios in the centre. Those are entirely accidental and non-universal. It is important to understand that since the central area corresponds to entirely randomly deciding the coarse-grained spin, the effective Hamiltonian (which would therefore have hardly \emph{anything} to do with the physics of the underlying system) would not even contain nearest neighbour terms. The central areas in Figs.(\ref{fig:densityplotlv2}b,c) thus correspond to ratios of two vanishing quantities. Similarly in Fig.(\ref{fig:rangenesslv2}a) the position of the peak not being exactly at 0 for small $||\Lambda||^2$ is exactly due to the accidental features in the centre of Figs.(\ref{fig:densityplotlv2}b).
	
	A slightly more practical lesson can be taken from Fig.(\ref{fig:densityplotlv2}c), where even for larger $\lambda$ multiple crossing of the 0-axis can be observed (i.e. the ``m-bodyness" ratio vanishes also for some smaller value of mutual information, compared to the value at the peak, when the ``rangeness" ratio is still large). This is also accidental, but teaches us that the proper metric to observe  is the saturation of the mutual information (corresponding to the peak) and not the vanishing of some particular coefficient in the Hamiltonian (which may be accidental).

	\subsection{Mutual information}\label{appendix:MI_Ising}
	Here we explicitly calculate the information-theoretic quantities studied in the main text for the case of the NN Ising model in 1D given by Eq.(\ref{eq:Ham_Ising}), with 
	a visible region of size $L_\V$ is coupled to a single hidden spin $\H=\{h\}$. The system is split into four regions [see Fig.(\ref{fig:system1Dising})] with their
	respective sizes satisfying $N=L_\V+2L_\B+2L_\E+L_\O$. We denote the spin variables in the three inner regions of the system by:
	\begin{IEEEeqnarray}{rCl}
		\V &=& \{v_1,v_2,\dotsc, v_{L_\V}\},\IEEEyessubnumber\\
		\B &=&
		\{b_{-L_\B},b_{-L_\B+1},\dotsc,b_{-1},b_1,b_2,\dotsc,b_{L_\B}\},\quad
		\IEEEyessubnumber\\
		\E &=& \{e_{-L_\E}, e_{-L_\E+1},\dotsc,e_{-1},e_1,e_2,\dotsc,e_{L_\E}\}.
		\IEEEyessubnumber
	\end{IEEEeqnarray}
	
	\subsubsection{Mutual information between the hidden degree of freedom and the environment}
	\label{appendix:MI_HE_Ising}	
	The mutual information can be calculated from Eq.(\ref{eq:appendix:MI_chain_rule}). Since $\H$ is a binary variable, the two entropies appearing in Eq.(\ref{eq:appendix:MI_chain_rule}) can be rewritten in terms of the binary entropy
	$h_2(p) := -p\log(p)-(1-p)\log(1-p)$:
	
	\begin{IEEEeqnarray}{rCl}
		H(\H) &=& h_2\left(P_\Lambda(h=1)\right),\\
		H(\H|\E) &=& \avg{h_2\left(P_\Lambda(h=1|\E)\right)}_\E,
	\end{IEEEeqnarray}
	with the conditional probability distribution:
	\begin{equation}
		P_\Lambda(\H|\E) =
		\frac{\sum_{\X\setminus\E} P_\Lambda(\H|\V)P(\X)}{P(\E)}.
	\end{equation}
	Thus, the mutual information is given by:
	\begin{equation}
	I_\Lambda(\H:\E) =
	h_2\left(P_\Lambda(h=1)\right)
	-\avg{h_2\left(P_\Lambda(h=1|\E)\right)}_\E.
	\label{eq:appendix:MI_calc}
	\end{equation}
	The relevant probability distributions, $P_\Lambda(h)$ and $P_\Lambda(h|\E)$, can be computed using transfer matrices (the result is always given in the limit $L_\O\rightarrow \infty$). For the former, we observe that:
		\begin{IEEEeqnarray}{rCl}
			P(\V) &=& \sum_{\B,\E,\O} P(\X)
			= \frac{1}{Z}\sum_{\B,\E,\O}\sum_{i=1}^N\matrixel{x_i}{T}{x_i} = \frac{e^{\K_\block[\V]}}{Z_\block},
			\IEEEnonumber
		\end{IEEEeqnarray}
	which implies that $P_\Lambda(h) = \sum_{\V}P(h|\V)P(\V) = P_\restblock(h)$,
	in the thermodynamic limit. We have already found $P_\restblock(h)$ in Eq.\ \eqref{eq:PLambdah} to be $\frac{1}{2}$, such that the first term in equation
	\eqref{eq:appendix:MI_calc} gives $h_2(1/2) = \log(2)$. The other relevant probability distribution is:
	\begin{equation}
	P_\Lambda(h|\E) = \sum_{\V}P(h|\V)P(\V|\E), 
	\end{equation}
	where $P(h|\V)$ is given by the RBM-ansatz Eq.(\ref{eq:appendix:RG_rule_genRBM_ansatz}) and to obtain $P(\V|\E)$ the two distributions $P(\V,\E)$ and $P(\E)$ need to be computed. In the thermodynamic limit $L_\O\to\infty$, we obtain by Eq.(\ref{eq:appendix:transfermatrixelement}):
	\begin{IEEEeqnarray}{rCl}
		P(\V,\E) &=& \sum_{\B,\O} P(\X)
		= \frac{1}{Z}\sum_{\B,\O}\sum_{i=1}^N\matrixel{x_i}{T}{x_i}
		\IEEEnonumber\\
		&=&
		\frac{1}{4}
		\left(1+v_1e_{-1}G(L_\B+1)\right)\left(1+v_2e_1G(L_\B+1)\right)
		\IEEEnonumber\\
		&&\quad\times\frac{e^{K\sum_{\nn{e,e'}}ee'}}{(2\cosh(K))^{2(L_\E-1)}}
		\frac{e^{\K_\block[\V]}}{Z_\block},
		\label{eq:appendix:PVE}
	\end{IEEEeqnarray}
	since  $\tanh(K)^m\to 0$ for $m\to\infty$ and finite $K$, and
	$Z=(2\cosh(K))^N$ in the thermodynamic limit.
	Similarly:
	\begin{IEEEeqnarray}{rCl}
	P(\E) &=& \sum_{\V,\B,\O} P(\X)
	= \frac{1}{Z}\sum_{\V,\B,\O}\sum_{i=1}^N\matrixel{x_i}{T}{x_i} =
	\IEEEnonumber\\
	&=& \frac{1}{4}\left(1+e_{-1}e_1G(L_\V+2L_\B+1)\right) \times\IEEEnonumber\\
	&&\quad\times\frac{e^{K\sum_{\nn{e,e'}}ee'}}{(2\cosh(K))^{2(L_\E-1)}}.
	\label{eq:appendix:PE} \\
P(\V|\E) &=&
\frac{
	\left(1+e_{-1}v_1G(L_\B+1)\right)\left(1+v_{L_\V}e_1G(L_\B+1)\right)
}{
1+e_{-1}e_1G(L_\V+2L_\B+1)
}\IEEEnonumber\\
&&\quad\times\frac{e^{\K_\block[\V]}}{Z_\block}, \\
	P_\Lambda(h|\E) &=& \frac{1}{2}\sum_{\V}P_\restblock(\V|h)
	\frac{
		1+e_{-1}v_1G(L_\B+1)
	}{
		1+be_{-1}e_1G(2(L_\B+1))
	} \times
	\IEEEnonumber\\
	&&\quad\times \left(1+v_{L_\V}e_1G(L_\B+1)\right),
\end{IEEEeqnarray}
where we recognized $P_\restblock(\V|h)$ from Eq.(\ref{eq:appendix:cumulant_dist}) and used the fact that:
\begin{IEEEeqnarray}{rCl}
G(L_\V+2L_\B+1) &=& \tanh(K)^{L_\V-1}G(2(L_\B+1))\IEEEnonumber\\
&=& bG(2(L_\B+1)).
\end{IEEEeqnarray}
By expanding the numerator, we can rewrite the above in terms of averages
$\avg{\cdot}_\restblock$, and using Eq.\ \eqref{eq:effbparams} we obtain:
\begin{IEEEeqnarray}{rCl}
	P_\Lambda(h|\E) &=&
	\frac{
		1+h(a_1e_{-1}+a_2e_1)G(L_\B+1)
	}{
		2\left(1+be_{-1}e_1G(2(L_\B+1))\right)
	}\IEEEnonumber\\
	&&\quad\times
	\frac{
		be_{-1}e_1G(2L_\B+2)
	}{
		2\left(1+be_{-1}e_1G(2(L_\B+1))\right)
	}.
\end{IEEEeqnarray}

$P_\Lambda(h|\E)$ only depends on the environment through $\{e_{-1},e_1\}$, so the sum over the remaining environment spins
in the average over $P(\E)$ can be performed explicitly, and we are left with an average over the marginal distribution: $P(e_{-1},e_1) = \frac{1}{4}\left(1+e_{-1}e_1 G(2L_\B+3)\right)$.
Finally, we can gather the results and obtain:
\begin{widetext}
	\begin{IEEEeqnarray}{rCl}\label{eq:d33}
		I_\Lambda(\H:\E) &=& \log(2)-\sum_{e_{-1},e_1}P(e_{-1},e_{1})
		h_2\left(\frac{
			1+(a_1e_{-1}+a_2e_1)G(L_\B+1)+be_{-1}e_1G(2L_\B+2)
		}{
			2\left(1+be_{-1}e_{1}G(2L_\B+2)\right)
		}\right).
	\end{IEEEeqnarray}
\end{widetext}

All dependence on $\Lambda$ is in the block parameters $a_1$, $a_2$ ($b$ is $\Lambda$-independent), calculated in Eqs.(\ref{eq:app_blockparams}).

\subsubsection{Mutual information between the visibles and the environment}
\label{appendix:MI_VE_Ising}
Equation (\ref{eq:MIbound_visibles}) states that the mutual information between the hiddens and the environment, $I_\Lambda(\H:\E)$, is bounded from above by the mutual information between the visibles and the environment, $I(\V:\E)$. We now compute the latter explicitly. By definition:
\begin{equation}
I(\V:\E)
= \sum_{\V,\E}P(\V,\E)\log\left(\frac{P(\V,\E)}{P(\V)P(\E)}\right),
\end{equation}
where all the probability distributions involved are already known, see Eqs.~(\ref{eq:appendix:PVE}) and (\ref{eq:appendix:PE}).
Observe that the expression inside the logarithm only depends on the four spins $e_{-1}$, $v_1$, $v_{L_\V}$ and $e_1$, such that the sum over all other spins can be performed explicitly. We obtain:
\begin{widetext}
	\begin{IEEEeqnarray}{rCl}
		I(\V:\E) &=& \sum_{\substack{e_{-1}, v_1,\\v_{L_\V}, e_1}}
		P(e_{-1}, v_1, v_{L_\V}, e_1)
		\log\left(\frac{
			\left(1+e_{-1}v_1G(L_\B+1)\right)\left(1+v_{L_\V}e_1G(L_\B+1)\right)
		}{
			1+e_{-1}e_1G(L_\V+2L_\B+1)
		}\right),\quad\\
		P(e_{-1}, v_1, v_{L_\V}, e_1) &=& \frac{1}{16}\left(1+e_{-1}v_1G(L_\B+1)\right)\left(1+v_1v_{L_\V}G(L_\V-1)\right)\left(1+v_{L_\V}e_1G(L_\B+1)\right).
	\end{IEEEeqnarray}
\end{widetext}

\subsection{The case of larger blocks}
For the case of $L_\V>2$ additional subtleties are present. These can be attributed to differently broken symmetries in the mutual information and in the effective Hamiltonian. 

\begin{figure*}
	\centering
	\includegraphics{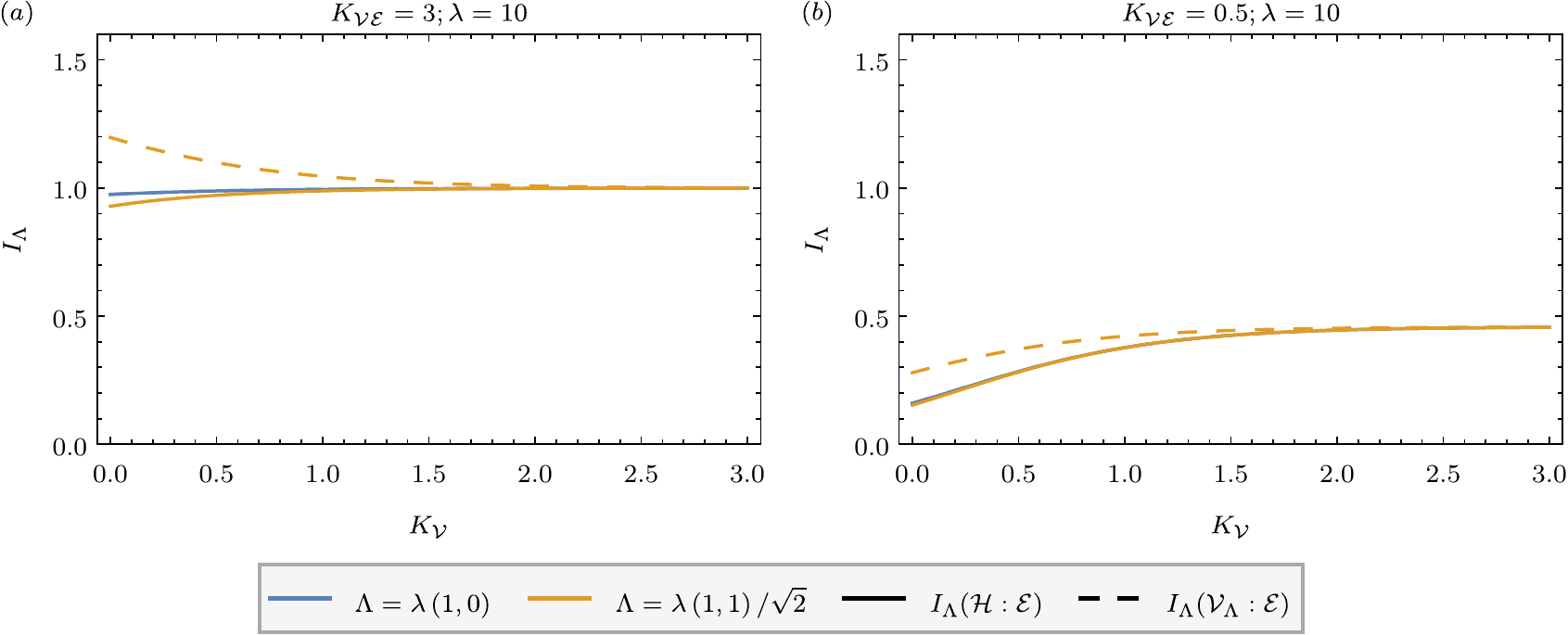}
	\caption{
		The mutual information $I_{\Lambda}(\H:\E)$ and $I(\V_\Lambda :\E)$ for decimation (blue) and majority rule (yellow) procedures in the 1D periodic toy model Eq.(\ref{eq:appendix:toymodel_1D_Ham}), obtained by coupling the environment spins $e_{1,2}$ in the model Eq.(\ref{eq:toymodel_1D_Ham}) with a coupling $K_\E = 1.5$ (the value shown). Compare with Fig.(\ref{fig:toymodel1d}), for which all other parameters are the same. Similarly, two parameter regimes are shown:
		(a) Strong coupling to the environment/ low temperature $K_{\V\E}$ (recall that the coupling constants contain a factor of  $\beta = 1/k_BT$)
		(b) Weak coupling $K_{\V\E}$.
		The solid lines differ from the dashed lines of the same colour by the mismatch $I(\V_{\Lambda}:\E | \H)$ (see the main text).
		Note that the introduction of $K_\E$ coupling, rendering the environment simply connected also in this 1D case, greatly reduces the difference between the two RG rules.	
	}
	\label{fig:appendix:toymodel1d}
\end{figure*}

On the level of interactions, the translation symmetry is explicitly broken by the Hamiltonian decomposition in Eq.(\ref{eq:Ham_decomp_intra_inter}) and subsequent cumulant expansion. This is not merely a feature of the method of evaluation, but rather a consequence of using a block-spin RG scheme: interactions of the spins in the same block are inherently treated differently from interactions of the spins from different blocks. However, the full translational symmetry may sometimes be effectively restored. This happens for instance in the case of a decimation, when for any block size $L_\V$ it does not matter which single spin exactly is chosen in the block -- the same effective Hamiltonian results.

When computing the mutual information, on the other hand, the full symmetry is not restored for $L_\V>2$.  The spins in the interior of the block are always coupled to the environment more weakly that the ones on the edges. Thus, we end up with two quantities, the renormalized Hamiltonian $\K'$ and the mutual information $I_\Lambda(\H:\E)$, which have different symmetry
properties. For example, for $L_\V=3$ in the 1D Ising case, from the point of view of mutual information we have two equivalent optimal solutions (coupling to left-most and right-most spins in the block), but it is intuitively clear that coupling to the center spin is equally good.

One important consequence is that the ``rangeness", for instance, is not necessarily a monotonic function of mutual information in the full parameter space (globally), but it is locally. Crucially though, any global maximum of mutual information corresponds to a global minimum of rangeness (but there could be additional equivalent solutions, just as the centre spin in the $L_\V=3$ decimation). The RSMI maximization is thus a \emph{sufficient} criterion for a good RG transformation, establishing it as a variational principle. Further investigation of these effects for larger coarse-graining blocks might prove useful (see also numerical results for the 2D Ising model case in the Supplementary Materials of Ref. \cite{Koch-Janusz2018}).


\section{Toy models}\label{appendix:toymodel}

\subsection{1D system}\label{appendix:toymodel1d}
To illustrate the influence of the environment $\E$ being simply connected or not, we modify the 1D toy model Eq.(\ref{fig:toymodel1d}) by introducing additional coupling $K_\E$ between the environment spins $e_{1,2}$, effectively making the system periodic (and thus the environment simply connected):

\begin{equation}
\K = K_{\V\E} (e_1v_1+v_2e_2) + K_\V v_1v_2 + K_\E e_1e_2.
\label{eq:appendix:toymodel_1D_Ham}
\end{equation}
This changes two things: on the one hand the visibles become more strongly coupled to each other. On the other, since the environment, for fixed $\V$, cannot now be thought of as being composed of two independent random variables $\E_1$ and $\E_2$, but rather a single one, the information about the environment copied into the visible spins $v_{1,2}$ is much more correlated. This has the effect of reducing the mismatch $I(\V_{\Lambda}:\E | \H)$. Indeed, as seen in Fig.\ \ref{fig:appendix:toymodel1d}, for the same values of all other parameters as in the non-periodic case of Fig.(\ref{fig:toymodel1d}), the discrepancy between the mutual information retained by the two coarse-graining rules is significantly decreased. Note though, that decimation still is (marginally) better.

\subsection{2D system}\label{appendix:toymodel2d}
As discussed in the main text, the situation in a two-dimensional system is qualitatively different. We consider the toy model with the Hamiltonian given by Eq.(\ref{eq:toymodel_2D_Ham}). Since all visibles couple to \emph{the same} environment $\E$, which is now a single variable $E \in \{-4,-3,\ldots,4\}$, in an identical fashion, each copies the same amount of information (at $K_\V=0$). Similar to the 1D case $\V_\Lambda$ captures more information about $\E$ if the coupling is more evenly distributed among the visibles. Additionally, with the connected environment, this has the effect of amplifying the shared information about $\E$ in each visible spin by averaging out the independent noise. While coupling to $\V_\Lambda$ always leads to more compression loss $I(\V_{\Lambda}:\E | \H)$, compared to decimation, the scale of the two effects is different such that in the 2D (and presumably also in higher-dimensional) case
the information gain when coupling to more visibles outweighs the compression loss, as seen in Fig.(\ref{fig:toymodel2d2}).

\section{Comparison to other definitions of RG optimality}\label{appendix:comparison}

The perfect action approach of Hasenfratz and Niedermayer~\cite{Hasenfratz1994}, as well as the approaches by Goldenfeld et al.~\cite{Goldenfeld1998} and Degenhard and Laguna~\cite{Degenhard2002} define an \emph{optimal} renormalization. In these works the goal of the optimization is fundamentally different from ours: the starting point is a continuum problem, which is replaced by an appropriate coarse-grained version, in order to numerically solve it. Optimality is then defined as minimizing the error of the solution with respect to the solution of the continuum problem. In contrast, we study a coarse-graining that captures the long-range physics while discarding short-range fluctuations, in a very general information-theoretic sense. It is, by construction, optimally compressing long-range information. We show that reduced complexity of the effective theory (i.e.~tractable Hamiltonian) \emph{is a consequence, even though it is not explicitly optimized for}. We also do not require a reference problem, such as the the continuum theory in the above cases.

More concretely, Ref.~\cite{Hasenfratz1994} discussed so-called perfect actions in the context of field theories. There exist lattice actions which give cutoff independent results on coarse-grained lattices -- they are perfect in the sense that they show no lattice artefacts. While the perfect action can be obtained from the RG flow, it is not unique; in particular, its range and the order of interactions depend on the details of the coarse-graining rule. Ref.~\cite{Hasenfratz1994} mentions the necessity of using this freedom to obtain actions that are as tractable (in exactly the same sense we use) as possible, to make the process practically feasible. This necessitates a second optimization, whose goal is to obtain a perfect action that can be easily approximated by as few coupling constants as possible. 
It is thus this second optimization that bears similarity to our problem. There are however two crucial differences: while in our approach this tractability is \emph{a consequence} of general principles, i.e.~\emph{a result}, in Ref.~\cite{Hasenfratz1994} it is explicitly optimized for by computing the actions for several different coarse-graining rules and tuning the RG procedure with respect to the range. Furthermore, while we work given a fixed initial system and consider an optimal RG step, Ref.~\cite{Hasenfratz1994} optimizes for tractability considering solutions of the fixed-point equations of motion.
Additionally, the coarse-grainings considered were restricted to a block-average based ones, and only optimized on as single parameter $P$ (corresponding to the magnitude $\Lambda$ in our notation), in effect exploring only a small subset of coarse-graining transformations.
That choice was justified in the particular example considered, but \emph{removing such choices altogether} is exactly the point of our work.

Ideas of coarse-graining as an alternative to sampling have also been applied to numerical methods for solving partial and stochastic differential equations~\cite{Goldenfeld1998,Degenhard2002,Hou2001,Degenhard2002b}.
In that case optimality is based on minimizing the difference between the solution of the coarse-grained and the coarse-grained solution of the continuous equation.
Effectively, this is a measure of the non-commutativity of coarse-graining and time-evolution.
Initially, a geometric coarse-graining was proposed~\cite{Goldenfeld1998} and the issue of finding approximations that make the resulting operators as short-ranged as possible was discussed~\cite{Hou2001}. Such a discussion could also benefit both from our generalization of the coarse-graining rules. Ref.~\cite{Degenhard2002} pointed out, though, that a geometric coarse-graining of partial differential equations cannot be fully satisfactory, since it does not take into account the specific dynamics. This is important, because the dynamics can change the relevance of the degrees of freedom compared to the equilibrium situation. Consequently a direct approach based on minimizing the error at later times was proposed. We do not consider dynamical problems, however a generalization of our approach would not involve a direct application of RSMI to coarse grain spatial degrees of freedom. This is inherently based on the notion that the important information to be compressed and preserved is the one pertaining to long \emph{spatial} length scales (which is correct in deriving effective theory in equilibrium). Correctly generalizing the compression intuition, in the spirit of Information Bottleneck approach~\cite{infbottle1}, would involve a procedure where the definition of the relevant information to be preserved also includes behaviour at longer \emph{time} scales. This is a very interesting question and another potential future research direction.

\end{document}